%% file: main.tex
\ifdefined\pdfobjcompresslevel
\fi
\documentclass[11pt]{article}

\usepackage[a4paper,margin=1in]{geometry}
\usepackage{amsmath,amssymb,amsthm}
\usepackage{mathtools}
\usepackage{enumitem}
\usepackage{booktabs}
\usepackage{graphicx}
\usepackage{float}
\usepackage{placeins}
\usepackage{authblk}
\usepackage{hyperref}
\hypersetup{hidelinks}

\usepackage{tikz}
\usetikzlibrary{arrows.meta,positioning,shapes.geometric}
\definecolor{hivepatch}{HTML}{E1D5E7}
\definecolor{hiveconnector}{HTML}{A8A8A8}
\definecolor{hiveedge}{HTML}{515151}
\definecolor{hiveedgedead}{HTML}{3F3F3F}
\definecolor{hivewhiteedge}{HTML}{F2EFE9}
\definecolor{hivewhiteedgedead}{HTML}{E5E0D8}
\definecolor{hivewhiteoutline}{HTML}{A9A39A}

\graphicspath{{images/}{ourImages/}}

\newtheorem{theorem}{Theorem}[section]
\newtheorem{lemma}[theorem]{Lemma}

\newtheorem{corollary}[theorem]{Corollary}
\theoremstyle{definition}
\newtheorem{definition}[theorem]{Definition}
\newtheorem{remark}[theorem]{Remark}

\newcommand{\Hive}{Hive}
\newcommand{\nHive}{$n$-Hive}
\newcommand{\singleQueenHive}{single-Queen \nHive{}}
\newcommand{\UCL}{Unbounded Two-Player Constraint Logic}

\title{Saving the Queen is hard:\\ 
Hive is EXPTIME-complete}

\author[1]{Francesco~Carboni\thanks{\texttt{carboni.2058986@studenti.uniroma1.it}}}
\author[1,2]{Daniele~Muscillo\thanks{\texttt{muscillo.2080466@studenti.uniroma1.it}}}

\affil[1]{Sapienza University of Rome}
\affil[2]{Orthogonal School, Elicsir Foundation}

\date{October 8, 2026}

\usepackage{subfiles}
\begin{document}

\subfile{Sec/Introduction}

\subfile{Sec/Preliminaries}

\subfile{Sec/TheReduction}

\subfile{Sec/RoyalRegions}

\subfile{Sec/BeetleStacks}

\subfile{Sec/End}

\begingroup\small
\makeatletter
\let\hive@originalthebibliography\thebibliography
\renewcommand{\thebibliography}[1]{%
  \hive@originalthebibliography{ACGMPT26}%
}
\let\hive@originalbibitem\bibitem
\renewcommand{\bibitem}[2][]{%
  \def\hive@thiskey{#2}%
  \def\hive@nytkey{alberti_et_al:LIPIcs.FUN.2026.2}%
  \ifx\hive@thiskey\hive@nytkey
    \hive@originalbibitem[ACGMPT26]{#2}%
  \else
    \hive@originalbibitem[#1]{#2}%
  \fi
}
\makeatother
\bibliographystyle{alpha}
\bibliography{references}
\endgroup
\end{document}

%% file: Sec/Introduction.tex
\maketitle

\begin{abstract}
\Hive{} is a boardless two-player strategy game in which the pieces
form the playing surface. We prove that deciding whether White can
force a win in \singleQueenHive{}, the generalization with exactly
one Queen Bee per player, is EXPTIME-complete, even when all pieces
are already placed.
Our hardness proof reduces from unbounded two-player Constraint
Logic: Grasshoppers encode edge orientations, while the One Hive
rule enforces the vertex constraints. We obtain the EXPTIME upper
bound by considering positions up to translation. The result
answers the single-Queen question raised by Andel and Rin. Keeping
the same inventory and movement rules, it also gives EXPTIME-completeness
for multi-Queen variants.
To the best of our knowledge, this is the first use of unbounded
two-player Constraint Logic to prove EXPTIME-completeness for a
natural problem.
\end{abstract}

\section{Introduction}
\label{sec:intro}

The computational complexity of games is studied through \emph{generalized games}: families of increasingly large instances for which one asks whether a puzzle has a solution or a player has a winning strategy. A recurring pattern is that one-player problems with polynomially bounded solution lengths are NP-complete, as in Tetris \cite{mithardnessgroup_et_al:LIPIcs.FUN.2026.32}, Sudoku \cite{yatoseta2003puzzles}, and natural generalizations of recent puzzles such as Pips and Strands \cite{alberti_et_al:LIPIcs.FUN.2026.2}. When the length of a puzzle solution is not polynomially bounded, PSPACE-completeness becomes a natural benchmark; examples include Sokoban, Rush Hour, and sliding-block puzzles \cite{demainehearn2008constraintlogic}. 
For two-player games, polynomially bounded play naturally leads
to PSPACE; Hex~\cite{reisch1981hex} and Othello~\cite{iwatakasai1994othello} are PSPACE-complete examples.
Games without a polynomial bound on play can instead be
EXPTIME-complete, as illustrated by generalized Chess
\cite{fraenkellichtenstein1981}, Checkers
\cite{robson1984checkers}, and Go under the rules
considered by Robson \cite{robson1983go}.

Constraint Logic provides a useful framework for proving such results
\cite{demainehearn2008constraintlogic,hearn2006thesis}. Its bounded and
unbounded one-player variants capture NP and PSPACE, respectively.
Bounded two-player Constraint Logic models polynomially long games and
has been used to prove PSPACE-completeness for games such as
Amazons and Konane. We use the unbounded two-player version, abbreviated 2CL, which is
EXPTIME-complete. Demaine, Hendrickson, and Lynch used unbounded 2CL
to establish EXPTIME-completeness within their general framework
for motion planning through gadgets \cite{dhl2020gadgets}.
Demaine, Gasarch, and Hajiaghayi distinguish this role in developing
hardness frameworks from applications to natural problems, noting
that unbounded 2CL had not yet been used to prove a natural problem
EXPTIME-complete \cite[p.~335]{dgh2026draft}.

We provide such an application for a generalization of \Hive{},
the boardless two-player strategy game designed by John Yianni
and published by Gen42 Games in 2001.
Its hexagonal pieces form a connected playing surface. Each player
controls a Queen Bee and other pieces with distinct movement rules.
The base game contains Queens, Ants, Beetles, Grasshoppers, and Spiders;
three official expansions add the Mosquito, Ladybug, and Pillbug.
Together they form the MLP ruleset used for rated play on Hivegame.com.%

A player loses when their Queen is surrounded on all six sides.
The rules combine local movement constraints with the global
\emph{One Hive} rule, which requires all pieces to remain connected.
\Hive{} also has an active community of tournament players, who also develop and discuss
unofficial insects and variants
\cite{worldhivecommunity2026,worldhivetournaments2026ehc,
boardgamegeek2026hivevariants}, and computer-playing agents,
developed through academic, educational, and open-source community
projects
\cite{out2017montecarlo,kampert2021mimicking,
orthogonalschool2025hivebot,mzinga2026}.
Despite this activity, the formal computational complexity of \Hive{} has
received comparatively little attention.

The boardless nature of \Hive{} makes it natural to generalize
the game by increasing the number of pieces.
No larger board is needed: additional copies of the pieces suffice,
while each player retains a single Queen.
We played such a version using two sets at the 2026 European
\Hive{} Convention \cite{worldhivetournaments2026ehc}.
The decision problem studied here concerns play from a given
position, rather than only from the empty starting position.

\paragraph{Related work.}
Andel and Rin \cite{andelrin2026hive} proved PSPACE-hardness
for a generalization of \Hive{} with an unrestricted number
of Queens per player, where surrounding any one Queen causes
that player to lose. They explicitly note that allowing multiple
Queens is arguably less natural than retaining a single Queen per
player, and identify the latter as more faithful to the spirit of
the original game.
They also discussed a variant in which
all of a player's Queens must be surrounded.
They left the exact complexity open and conjectured
EXPTIME-completeness for their main variant.
The complexity of the more natural single-Queen version remained
open as well. Their reduction uses all five base-game piece types and makes
essential use of multiple Queens to create selfmate hexes. They note
that Ants are used only rarely and may be replaceable, and ask whether
hardness can also be established without Spiders.

\paragraph{Our contribution.}
We study \singleQueenHive{}, in which each player has exactly one
Queen Bee, as in the standard game, while the numbers of the remaining
pieces may grow with the instance size. We prove that its winning-strategy
decision problem is EXPTIME-complete, by a polynomial-time reduction
from \UCL{} and a matching upper bound. Hardness holds even with all
pieces placed and for every inventory variant with enough Grasshoppers,
Beetles, and Mosquitoes to realize the construction.

We also consider multi-Queen variants with the two loss conditions
studied by Andel and Rin, leaving the inventory and movement rules
unchanged (Definition~\ref{def:multi-queen-variants}). Both are
EXPTIME-complete: single-Queen instances belong to both variants,
and the same EXPTIME upper bound applies to all three models.

Our construction also reduces the required piece set. Its active gadgets
use only Queen Bees, Beetles, Grasshoppers, and Mosquitoes, and in
particular require neither Ants nor Spiders. Additional inventory pieces
can be kept immobile as described in Remark~\ref{rem:unused-inventory}.
Moreover, the Mosquitoes can be eliminated while preserving essentially
the same gadget structure. The resulting construction requires
substantially more Beetles and a more involved correctness argument,
but its active gadgets use only Queen Bees, Beetles, and Grasshoppers.
% V4-CHECK EXTENSIONS: the Mosquito-free construction and correctness
% argument are now supplied in BeetleStacks.tex. Verify the replacement
% invariants and the royal-region arguments inherited by this extension.

%% file: Sec/Preliminaries.tex
\section{Preliminaries}
\label{sec:preliminaries}

\subsection{The target problem: \texorpdfstring{\singleQueenHive{}}{single-Queen n-Hive}}
\label{sec:hive}

\Hive{} has no board: its hexagonal pieces form the playing surface. To encode a position, we
regard them as occupying cells of the infinite hexagonal grid $\mathcal{H}$. The cell-adjacency
graph of this grid is the triangular lattice. Each cell has six neighbours, and each straight
lattice step points in one of six directions. The grid imposes no boundary on play.

Let $\mathcal{C}=\{W,B\}$ denote White and Black. The piece types of modern \Hive{} with its
official expansions are
\[
  \mathcal{T}=\{Q,A,\mathrm{Bt},G,S,M,P,L\},
\]
standing for Queen Bee, Ant, Beetle, Grasshopper, Spider, Mosquito, Pillbug, and Ladybug.
We use the movement rules of the base game and its official expansions.
An instance explicitly lists a common inventory of $n$ pieces for each
player. If $m_X$ is the number of pieces of type $X$ in that inventory, then
$\sum_{X\in\mathcal T}m_X=n$ and $m_Q=1$. White and Black have the same multiplicity $m_X$ for
every type $X$; pieces not yet placed remain in their respective reserves.

The simulation and punishment mechanisms use Queen Bees, Grasshoppers,
Beetles, and Mosquitoes. We recall their movement rules below.
Additional inventory pieces, when required, are placed in an
immobile appendage described in Remark~\ref{rem:unused-inventory}. We write $X^C$ for a piece of type $X$ controlled by $C$, so $G^W$ is a White
Grasshopper and $Q^B$ the Black Queen. Subscripts distinguish copies, and $*$ is a colour
wildcard used when a piece's colour is irrelevant. 
In the single-Queen construction, $Q^C$ is player $C$'s Queen and
$F^C$ its unique initially empty neighbour, the \emph{mate cell}
(Section~\ref{sec:royal}). 
Pieces may climb onto others, forming stacks. A \emph{stack} is a finite
nonempty sequence $\sigma=(p_1,\ldots,p_k)$, listed from bottom to top.
It is legal if every $p_i$ with $i>1$ is a Beetle or a Mosquito.
A stack of length one is \emph{flat}. Only the top piece can move or
use an ability; covered pieces cannot act. For placement purposes,
a stack has the colour of its top piece.

Players alternate turns. A player may place a piece from their reserve,
move an uncovered piece they control, or use an available special ability.
Pieces are neither captured nor removed.

A player cannot move pieces before placing their Queen, which must be
placed by their fourth turn. After each player's first
placement, every piece they place must touch at least one piece of their
own colour and none of the opponent's colour.

Movement is also subject to two general rules, illustrated in
Figure~\ref{fig:rules-global}(b)--(d).

\begin{description}[itemsep=2pt]
\item[One Hive rule.]
We also refer to the set of occupied cells as the \emph{hive}; it must remain connected
throughout a move. In particular, a ground-level
piece cannot leave its cell if doing so disconnects the hive, even when its destination would
reconnect it. This applies to jumps as well as slides.
We call a ground-level piece a \emph{cut piece} if removing it
disconnects the hive.

\item[Freedom to Move rule.]
A sliding piece must fit through every gap it crosses and remain in
contact with the hive. For a ground-level slide from $x$ to an adjacent
empty cell $y$, exactly one of the two common neighbours of $x$ and $y$
must be empty: one leaves room to pass, while the other maintains contact.
For a Beetle step, let $h_x$ and $h_y$ be the source and destination
heights before the move, with height zero for an empty cell.
The travel level is $\max\{h_x,h_y+1\}$; the two common neighbours
block the step if both reach that level.
Grasshopper jumps, including those copied by Mosquitoes, are not
subject to this sliding restriction.
\end{description}

If no legal action is available, the player passes.
Every piece is already placed in the instances produced by our reduction,
so placement does not occur there.

We now describe how the four piece types used by the simulation move;
Figure~\ref{fig:rules-movement} summarizes their available destinations.
\begin{description}[leftmargin=*,itemsep=2pt,font=\normalfont\itshape]
\item[Queen Bee.]
A Queen Bee slides one cell to an adjacent empty cell.

\item[Grasshopper.]
A Grasshopper jumps rather than slides. It chooses one of the six lattice directions, jumps
over a nonempty contiguous run of occupied cells, and lands on the first empty cell beyond
it. It cannot jump in a direction whose immediately adjacent cell is empty.

\item[Beetle.]
A Beetle moves one cell per turn. On the ground it slides to an adjacent empty cell or climbs
onto an adjacent occupied one. Above the hive it slides between neighbouring stack tops and
may descend into an adjacent empty cell. Until uncovered, a piece it covers cannot move,
use an ability, or supply an ability to an adjacent Mosquito.

\item[Mosquito.]
At ground level, a Mosquito may copy an ability of a neighbouring
piece at the start of its move, regardless of colour. Only the top
piece of a neighbouring stack supplies an ability, and a Mosquito
cannot copy another Mosquito, including an elevated one.
Consequently, a ground-level Mosquito whose occupied neighbours are
all Mosquitoes copies no movement ability and has no legal move of its own.
Above ground level it moves as a Beetle until it returns to the ground.

\end{description}
\paragraph{The Pillbug and the preceding turn.}
Instead of moving itself, an uncovered Pillbug may relocate an adjacent
unstacked piece of either colour: it lifts that piece onto itself
and places it in an empty adjacent cell, subject to One Hive and
the ability's movement restrictions \cite{hivepillbugv4}.
It cannot relocate the piece moved or placed on the immediately
preceding turn. A piece relocated by this ability cannot move,
use an ability, or be relocated again on the next turn.
The same restrictions apply when a Mosquito copies the Pillbug.
Thus the game state records the last piece moved or placed,
together with a flag indicating relocation by this ability;
after a pass, this record is empty.

\begin{figure}[tbp]
  \centering
  \begin{minipage}[t]{0.235\linewidth}
    \centering
    \includegraphics[height=2.9cm,keepaspectratio]{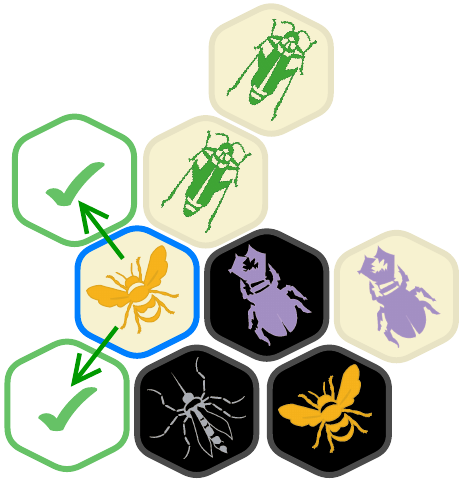}
    \par\smallskip{(a) Queen Bee}
  \end{minipage}\hfill
  \begin{minipage}[t]{0.235\linewidth}
    \centering
    \includegraphics[height=2.9cm,keepaspectratio]{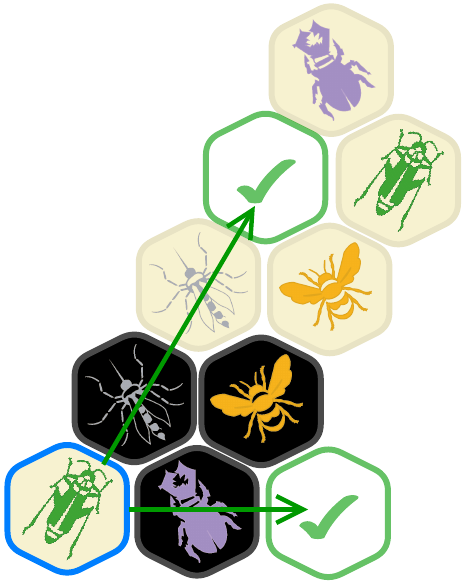}
    \par\smallskip{(b) Grasshopper}
  \end{minipage}\hfill
  \begin{minipage}[t]{0.235\linewidth}
    \centering
    \includegraphics[height=2.9cm,keepaspectratio]{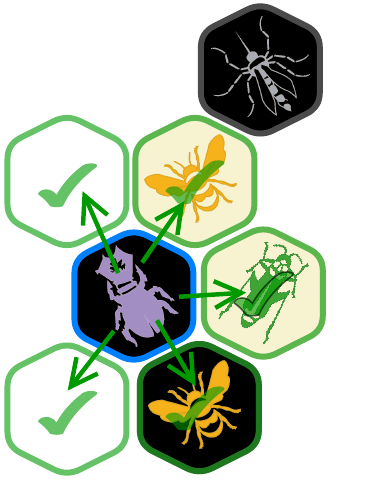}
    \par\smallskip{(c) Beetle}
  \end{minipage}\hfill
  \begin{minipage}[t]{0.235\linewidth}
    \centering
    \includegraphics[height=2.9cm,keepaspectratio]{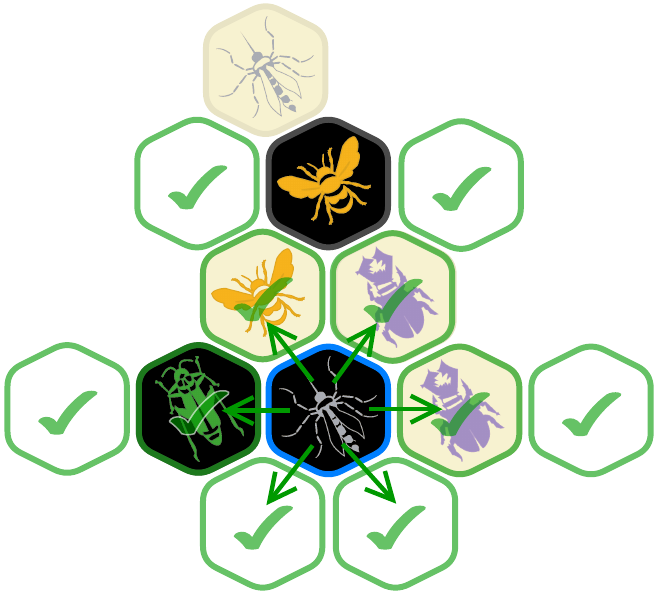}
    \par\smallskip{(d) Mosquito}
  \end{minipage}
  \caption{The four piece types used by the simulation.
  Arrows and marked cells show their available destinations. A Mosquito
  selects an ability from a neighbouring non-Mosquito top piece at the
  start of its move.}
  \label{fig:rules-movement}
  \label{fig:rules-queen}
  \label{fig:rules-grasshopper}
  \label{fig:rules-beetle}
  \label{fig:rules-mosquito}
\end{figure}

\begin{figure}[H]
  \centering
  \begin{minipage}[t]{0.235\linewidth}
    \centering
    \includegraphics[width=\linewidth,height=2.8cm,keepaspectratio]
      {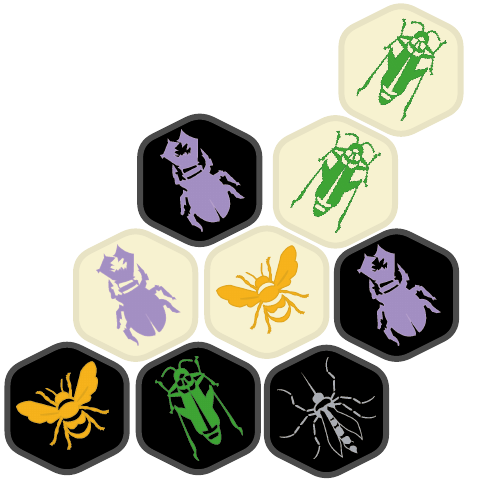}
    \par\smallskip (a) Winning condition
  \end{minipage}\hfill
  \begin{minipage}[t]{0.235\linewidth}
    \centering
    \includegraphics[width=\linewidth,height=2.8cm,keepaspectratio]
      {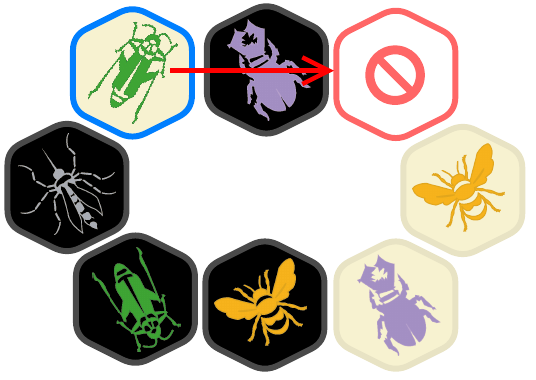}
    \par\smallskip (b) One Hive
  \end{minipage}\hfill
  \begin{minipage}[t]{0.235\linewidth}
    \centering
    \includegraphics[width=\linewidth,height=2.8cm,keepaspectratio]
      {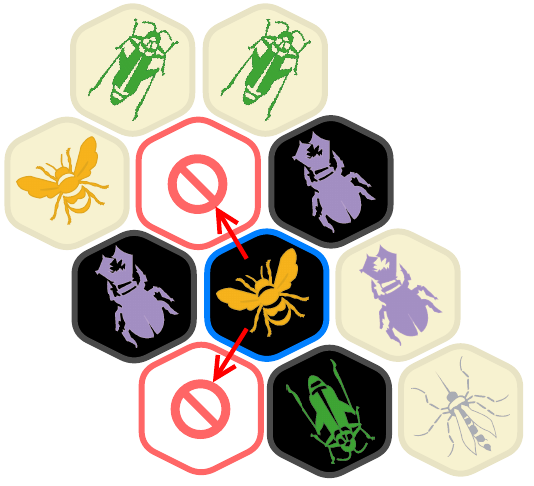}
    \par\smallskip (c) Freedom to Move
  \end{minipage}\hfill
  \begin{minipage}[t]{0.235\linewidth}
    \centering
    \includegraphics[width=\linewidth,height=2.8cm,keepaspectratio]
      {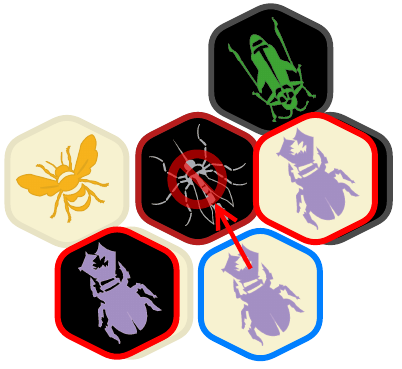}
    \par\smallskip (d) Elevated gate
  \end{minipage}
  \caption{Winning and movement restrictions. (a) A surrounded Queen loses.
  (b) One Hive forbids disconnecting moves. (c) Freedom to Move blocks a
  slide through a closed gap. (d) An elevated gate blocks a Beetle slide over the hive.}
  \label{fig:rules-global}
  \label{fig:rules-winning}
\end{figure}

\begin{definition}
\label{def:position}
Given the common inventory $\mathbf m=(m_X)_{X\in\mathcal T}$, a
\emph{\singleQueenHive{} position} is a finite partial function
$P:\mathcal H\rightharpoonup\mathrm{Stacks}$ such that:
\begin{enumerate}[label=\textup{(\arabic*)}]
\item every stack of $P$ is legal;
\item for each $C\in\{W,B\}$ and $X\in\mathcal T$, at most $m_X$ pieces of type $X$
      controlled by $C$ occur in $P$;
\item the support
      $\operatorname{supp}(P)=\{h\in\mathcal H:P(h)\text{ is defined}\}$ is connected in the
      triangular lattice;
\item for each $C\in\{W,B\}$, exactly one piece of $P$ has type $Q$ and colour $C$.
\end{enumerate}
Pieces in the inventory but absent from $P$ remain in reserve.
\end{definition}

The inventory condition $m_Q=1$ enforces the single-Queen model;
the last condition above requires both Queens to be on the hive.
The decision question concerns play from the specified position.

A position describes the arrangement of pieces. The full \emph{game state}
also includes the player to move and the preceding-turn record
$r=(p,b)$: $p$ identifies the last piece moved or placed, and $b$
indicates whether it was relocated by the Pillbug ability.
A pass gives the empty record. This is the only additional memory
required by the rules described above. In the simulated part of
our construction the record has no effect, since no Pillbug
ability is available.

A position contains at most $2n$ pieces, so its connected support has
diameter at most $2n-1$ in the triangular lattice. After translation,
the cells, ordered stack contents, player to move, and record $r$
can all be encoded using polynomially many bits in $n$.

\paragraph{Winning condition, draws, and infinite play.}
A player loses as soon as all six cells adjacent to their Queen are
occupied, regardless of the surrounding pieces' types, colours, or heights.
This condition is illustrated in Figure~\ref{fig:rules-winning}(a).
If both Queens are surrounded, the outcome is a draw.
These conventions also determine the outcome of an already terminal input.

For example, suppose six pieces form a straight
line in the order
\[
M^W,M^B,Q^W,Q^B,M^W,M^B
\]
and both reserves are empty: all six pieces are immobile and
both players pass indefinitely.
More generally, play may continue forever without either player winning. We
treat every such infinite play as a draw; a draw is not a forced win for
White.\\
The six-piece example also illustrates the
two-Mosquito tail, a \emph{pinning technique} used repeatedly throughout the proof.
Mosquitoes $M_1,M_2$ extending a straight line from a flat non-Pillbug piece
$p$, with no other occupied neighbours, pin all three pieces: $p$ and $M_1$
are cut pieces, $M_2$ cannot copy $M_1$, and $M_1$ cannot obtain from $p$ a
relocation ability that moves $M_2$.

\begin{definition}
\label{def:decision}
An instance of \textit{single-Queen $n$-Hive} specifies a unary size parameter
$1^n$, an explicitly listed common inventory
$\mathbf m=(m_X)_{X\in\mathcal T}$ satisfying
$\sum_{X\in\mathcal T}m_X=n$ and $m_Q=1$, a position $P$ for that
inventory, the player $C\in\{W,B\}$ to move, and the preceding-turn
record $r$. Both players receive the same inventory. The corresponding decision problem asks whether White has a strategy
that forces a win from the specified game state.
\end{definition}

\begin{definition}
\label{def:multi-queen-variants}
The two \emph{multi-Queen variants} keep the inventory and movement
rules of \singleQueenHive{}, but allow any number $m_Q\ge1$ of
Queens per player, all already on the hive. Following the two loss conditions considered by Andel and Rin~\cite{andelrin2026hive},
a player loses when any one of their Queens is surrounded in the first variant.
In the second variant, a player loses when all of their Queens are surrounded. If both players satisfy the loss condition, the outcome
is a draw. All other rules and the forced-win question remain unchanged.
\end{definition}

\subsection{The source problem: \texorpdfstring{\UCL}{Unbounded Two-Player Constraint Logic}}
\label{sec:2cl}

We reduce from unbounded two-player Constraint Logic
\cite{demainehearn2008constraintlogic}. A \emph{constraint graph} is an undirected graph
$\mathcal{G}=(V,E)$ whose edges have weights in $\{1,2\}$ and whose vertices have minimum inflow
requirements. Here every vertex requires inflow at least $2$. A \emph{configuration} orients
every edge; it is \emph{legal} if the total weight directed into each vertex meets that
requirement. Reversing one edge is legal when the resulting configuration is legal.

We use AND/OR constraint graphs. An \emph{AND vertex} has two weight-$1$ edges (red) and one
weight-$2$ edge (blue). Its blue edge can point outward if and only if both red edges point
inward. An \emph{OR vertex} has three blue edges; any one can point outward if and only if at
least one of the other two points inward. Figure~\ref{fig:2cl-vertices} illustrates both
constraints. Every vertex has degree $3$.

\begin{figure}[h]
\centering
\begingroup
  \definecolor{clred}{RGB}{255,0,0}
  \definecolor{clblue}{RGB}{0,0,255}
  \tikzset{
    cl outline/.style={line cap=round,line join=round},
    cl red/.style={cl outline,draw=clred,line width=3.2pt},
    cl blue/.style={cl outline,draw=clblue,line width=5pt},
    cl core/.style={line cap=round,line join=round}
  }

  % Edge and chevron are drawn in the SAME two passes:
  % first both coloured outlines, then both white/black interiors.
  % This removes the coloured seam at the arrow--edge junction.
  % Arguments: weight, owner colour, edge angle, arrow position,
  %            direction (+1 outward, -1 inward).
  \def\CLEdge#1#2#3#4#5{%
    \ifnum#1=1\relax
      \def\CLStyle{cl red}%
      \def\CLCoreWidth{1.4pt}%
      \def\CLHeadLength{.29}%
      \def\CLHeadWidth{.17}%
    \else
      \def\CLStyle{cl blue}%
      \def\CLCoreWidth{2.2pt}%
      \def\CLHeadLength{.32}%
      \def\CLHeadWidth{.20}%
    \fi
    \begin{scope}[rotate=#3]
      \draw[\CLStyle]
        (0,0) -- (1.25,0)
        ({#4-#5*\CLHeadLength},\CLHeadWidth)
          -- (#4,0)
          -- ({#4-#5*\CLHeadLength},-\CLHeadWidth);
      \draw[cl core,draw=#2,line width=\CLCoreWidth]
        (0,0) -- (1.25,0)
        ({#4-#5*\CLHeadLength},\CLHeadWidth)
          -- (#4,0)
          -- ({#4-#5*\CLHeadLength},-\CLHeadWidth);
    \end{scope}%
  }

  \def\CLVertex{%
    \filldraw[fill=gray!65,draw=gray!80,line width=.4pt]
      (0,0) circle (3.1pt);
  }

  % Equal-width panels spread the drawings across the available line.
  % Scale affects the geometry; label font sizes remain unchanged.
  \def\CLPanel#1#2{%
    \begin{minipage}[t]{.245\linewidth}
      \centering
      \begin{tikzpicture}[scale=1.15]
        \path[use as bounding box]
          (-1.30,-1.35) rectangle (1.30,1.40);
        #2
        \CLVertex
        \node at (0,-1.15) {#1};
      \end{tikzpicture}%
    \end{minipage}%
  }

  \noindent
  \CLPanel{(a) AND}{%
    \CLEdge{1}{white}{210}{.29}{-1}
    \CLEdge{1}{black}{330}{.29}{-1}
    \CLEdge{2}{white}{ 90}{.85}{ 1}
  }\hfill
  \CLPanel{(b) AND}{%
    \CLEdge{1}{white}{210}{.29}{-1}
    \CLEdge{1}{black}{330}{.89}{ 1}
    \CLEdge{2}{white}{ 90}{.29}{-1}
  }\hfill
  \CLPanel{(c) OR}{%
    \CLEdge{2}{black}{210}{.29}{-1}
    \CLEdge{2}{white}{330}{.89}{ 1}
    \CLEdge{2}{black}{ 90}{.85}{ 1}
  }\hfill
  \CLPanel{(d) OR}{%
    \CLEdge{2}{black}{210}{.29}{-1}
    \CLEdge{2}{white}{330}{.29}{-1}
    \CLEdge{2}{black}{ 90}{.85}{ 1}
  }
\endgroup
\caption{Four legal orientations of AND and OR vertices. Red and blue outlines
indicate edge weights $1$ and $2$; white and black interiors indicate control
by White and Black. Arrowheads show the orientations. Each vertex requires
inflow at least $2$.}
\label{fig:2cl-vertices}
\end{figure}
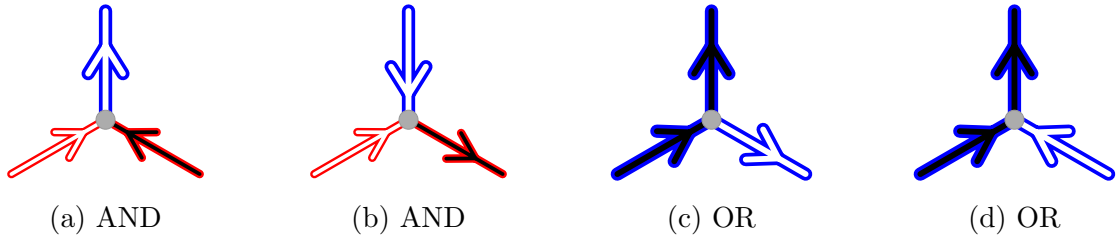

Each edge is controlled by exactly one player: $E_W$ and $E_B$ partition $E$. On their turn,
player $C$ may reverse one edge of $E_C$, provided the configuration remains legal.
Each player has a target edge, $e_W\in E_W$ and $e_B\in E_B$, and wins as soon as they
reverse their own target. Edges may be reversed arbitrarily many times. An infinite play
in which neither player wins is a draw, not a forced win for White.

\begin{definition}
\label{def:2cl}
An instance of \UCL{} specifies an AND/OR constraint graph $\mathcal{G}$,
an initial legal configuration $\mathcal{O}$, a partition $E_W,E_B$ of its
edges, target edges $e_W\in E_W$ and $e_B\in E_B$, and the player
$C_0\in\{W,B\}$ to move. The corresponding decision problem asks whether
White has a strategy that forces a win from the specified game state.
\end{definition}

\begin{remark}[Pass convention]
\label{rem:no-pass}
We allow players to pass even when a legal edge reversal is available.
This convention is compatible with the reduction of Demaine and Hearn
from $G_6$, where passing is explicitly permitted
\cite[\S6.2]{demainehearn2008constraintlogic}.
A pass in $G_6$ can also be realized by flipping a fresh variable
belonging to the player and occurring only in a tautological clause,
leaving that player's strategy unchanged.
In 2CL, such a move can be implemented by a separate component
containing no target edge: join two free edge terminators
\cite[Figure~5(a)]{demainehearn2008constraintlogic}
and assign all edges of the resulting component to the same player.
Repeatedly reversing the joining edge then represents passing.
During variable play, these moves reproduce the passes of $G_6$;
during locking and formula evaluation, passing neither answers
a forcing threat nor accelerates a winning path, so the same strategy
and timing arguments apply.
We adopt explicit passes for convenience and implement them in
\Hive{} using the gadgets of Section~\ref{sec:pass}.
\end{remark}
% V4-CHECK SOURCE: identify how the Demaine--Hearn construction represents
% an unchanged assignment and check the effect of passes in its locking
% and winning phases. The simulation is not globally one move for one move:
% locking and formula evaluation use sequences of edge reversals.
% This concerns the precise source theorem below, not the Hive pass gadget.

\begin{theorem}
\label{thm:2cl-exptime}
\UCL{} is EXPTIME-complete, even for planar AND/OR constraint graphs.
\end{theorem}

For the reduction to \Hive{}, we must represent legal edge reversals,
enforce the AND and OR constraints, and implement passes.

\subsection{Overview of the proof}
\label{sec:overview}

We prove that \singleQueenHive{} is EXPTIME-complete. For hardness, we reduce
from the pass-enabled 2CL game defined in Section~\ref{sec:2cl};
its EXPTIME-completeness is stated in Theorem~\ref{thm:2cl-exptime}.
The reduction has two interacting layers: at ground level, Grasshoppers
encode edge orientations and the One Hive rule enforces the AND and OR
constraints. Above the hive, Mosquitoes acting as Beetles provide a response to moves that
leave this simulation.

\paragraph{Preparing the construction.}
For the source graph $\mathcal{G}$, let $N=|V(\mathcal{G})|+|E(\mathcal{G})|$.
We map an orthogonal drawing to the triangular lattice and scale it
by $s$ to separate gadgets and reserve their jump rays and landings
(Section~\ref{sec:embedding}). Each source vertex is replaced by an AND or OR gadget that uses the One Hive rule to enforce the corresponding inflow constraint. One straight segment of each edge
carries its Grasshopper gadget; the remaining segments extend the
incident vertex structures. Two outer-face gadgets implement source
passes without changing any encoded orientation (Section~\ref{sec:pass}).

\paragraph{Faithful simulation.}
Before either target reversal, a \emph{faithful configuration} has
edge and pass Grasshoppers at designated resting cells and all other
pieces in their prescribed positions, apart from neutralized lateral
stacks and harmless royal-guard moves (Section~\ref{sec:royal}).
A \emph{faithful move} represents a legal source-edge reversal or
pass; harmless royal-guard moves also represent passes.
A lateral lid move and its immediate prescribed countercover form a
\emph{neutral interruption}. Any other move from a faithful
configuration is a \emph{deviation}; neutral interruptions and their
prescribed replies are excluded.

A \emph{faithful simulation} consists of faithful moves and completed
neutral interruptions up to the first target reversal. Each completed
interruption uses one turn per player but changes no encoded source
orientation. After the prescribed countercover, the source position
and player to move are as they were before the interruption, so the
source strategy can be resumed.

\paragraph{The route to victory and punishing deviations.}
A target reversal sends its Grasshopper into a \emph{royal region}
containing the opponent's Queen. The prescribed traversal has
$\Theta(s)$ rows of length $\Theta(s)$ and requires $\Omega(s^2)$
moves before the final winning move (Section~\ref{sec:royal}).
Traps release opposing Mosquito rovers, which follow the connected
edge--connector network to the deviator's Queen in $O(N^2s)$ moves
(Sections~\ref{sec:deviations-traps} and~\ref{sec:correctness-assembly}).
For a first deviation before either target reversal, the deviator
still faces the full traversal. Choosing $s=\Theta(N^3)$ with a
sufficiently large constant makes punishment faster; the global
argument accounts for intervening responses and excludes a faster
win or an immediate draw.
% V4-CHECK TIMING: P < W compares a punishment with a full royal traversal.
% It does not bound the time remaining after a traversal has already begun.

\paragraph{Correctness and Mosquito removal.}
The full strategy correspondence, including deviations and infinite
play, is proved in Section~\ref{sec:reduction-correctness}. The
construction uses $O(N^2s+s^2)$ pieces and is computable in polynomial
time, giving EXPTIME-hardness. Membership in EXPTIME is proved
separately from the finite game graph of positions normalized up to
translation (Lemma~\ref{lem:hive-membership-final}).
Finally, we strengthen the reduction by removing Mosquitoes
(Section~\ref{sec:beetle-stacks}): most of them serve only as
immobile structural supports and we replace these supports with
Beetle stacks. This yields the same
hardness result using only Queens, Grasshoppers, and Beetles.

\paragraph{Terminology.}
We introduce several terms used repeatedly throughout the reduction,
including terminology specific to \Hive{}.
The adjective \emph{royal} refers to pieces or structures belonging
to the royal regions introduced in Section~\ref{sec:royal}.
A piece is \emph{pinned} when it cannot legally move in the position
under consideration; this does not by itself prevent abilities
that leave the piece in place.
A \emph{lid} is the top piece of a covering stack.
To \emph{countercover} a displaced lid means to place an opposing
piece on top of it, preventing it from acting while covered.
A \emph{guard} is a piece assigned to respond to a deviation and
neutralize its effect.
A \emph{rover} is a piece released to carry out punishment by
travelling above the hive.

%% file: Sec/TheReduction.tex
\section{The reduction: embedding and gadgets}
\label{sec:reduction-gadgets}

We describe how to transform
a 2CL instance into a \Hive{} position.

Many pieces used solely as structural supports may belong to either player;
we state explicitly when ownership is irrelevant. In the figures, we draw
such pieces in the colour assigned to their structure: purple for vertex
gadgets and gray for connectors.

\subsection{Geometric embedding}
\label{sec:embedding}

We obtain the embedding of $\mathcal{G}$ in Hive by drawing $\mathcal{G}$ orthogonally, transferring the drawing to
the triangular lattice of Hive-cell centres, and stretching it to make room
for the gadgets. Since $\mathcal{G}$ is planar and has maximum degree three,
Biedl and Kant~\cite{biedlKant1998} give an orthogonal drawing of each
connected component on an $O(N)\times O(N)$ square grid, with at most two bends
per edge. Placing the components side by side preserves these bounds. Every
edge then has length $O(N)$, and all edges together have total length $O(N^2)$.

Let $a,b$ be unit vectors of the triangular lattice, with $b$ obtained by
rotating $a$ counterclockwise through $60^\circ$. The invertible linear map
$T(x,y)=xa+yb$ is an isomorphism from the point lattice $\mathbb Z^2$ to
$\mathbb Za+\mathbb Zb$. It is also a graph isomorphism from the square grid to the subgrid formed
by the $a$- and $b$-directions, before adding the third direction $b-a$.
As an orientation-preserving map, it preserves
planarity and the cyclic order of edges; because $a$ and $b$ are unit vectors,
it also preserves grid-path lengths. Thus we obtain in polynomial time a
planar drawing of $\mathcal{G}$ on the triangular lattice, with at most two bends per
edge and total length $O(N^2)$. Figure~\ref{fig:embedding-stages} shows the
steps.

\begin{figure}[H]
\centering
\begin{tikzpicture}[
  embedge/.style={
    draw=blue!65!black,
    line width=1.1pt,
    line cap=round,
    line join=round
  },
  embgrid/.style={draw=gray!32,line width=0.35pt},
  embbend/.style={fill=orange!85!black},
  every node/.style={font=\scriptsize}
]

  \node[font=\footnotesize] at (1.875,3.38)
    {(a) Graph $\mathcal{G}$};
  \node[font=\footnotesize] at (5.875,3.38)
    {(b) Orthogonal drawing};
  \node[font=\footnotesize] at (9.875,3.38)
    {(c) Linear transformation};
  \node[font=\footnotesize] at (13.875,3.38)
    {(d) Triangular lattice};

  % (a) A small planar subcubic graph.
  \draw[embedge]
    (0.45,1.55)--(1.875,2.63)--(3.30,1.55)
    --(1.875,0.40)--cycle;
  \draw[embedge] (0.45,1.55)--(3.30,1.55);

  \fill (0.45,1.55) circle[radius=0.06cm];
  \fill (3.30,1.55) circle[radius=0.06cm];
  \fill (1.875,2.63) circle[radius=0.06cm];
  \fill (1.875,0.40) circle[radius=0.06cm];

  \node[anchor=east,xshift=-4pt] at (0.45,1.55) {$u$};
  \node[anchor=west,xshift=4pt] at (3.30,1.55) {$v$};
  \node[above,yshift=3pt] at (1.875,2.63) {$p$};
  \node[below,yshift=-3pt] at (1.875,0.40) {$q$};

  % (b) Orthogonal drawing.
  \begin{scope}[shift={(4,0)}]
    \begin{scope}[shift={(0.62,0.40)}]
      \begin{scope}[x=0.55cm,y=0.55cm]

        \foreach \i in {0,...,4} {
          \draw[embgrid] (\i,0)--(\i,4);
          \draw[embgrid] (0,\i)--(4,\i);
        }

        \draw[embedge] (0,2)--(0,4)--(2,4);
        \draw[embedge] (2,4)--(4,4)--(4,1);
        \draw[embedge] (0,2)--(2,2)--(2,1)--(4,1);
        \draw[embedge] (0,2)--(0,0)--(2,0);
        \draw[embedge] (2,0)--(4,0)--(4,1);

        \foreach \x/\y in {0/4,4/4,0/0,4/0,2/2,2/1}
          \fill[embbend] (\x,\y) circle[radius=0.045cm];

        \foreach \x/\y in {0/2,4/1,2/4,2/0}
          \fill (\x,\y) circle[radius=0.06cm];

        \node[anchor=east,xshift=-4pt] at (0,2) {$u$};
        \node[anchor=west,xshift=4pt] at (4,1) {$v$};
        \node[above,yshift=3pt] at (2,4) {$p$};
        \node[below,yshift=-3pt] at (2,0) {$q$};

      \end{scope}
    \end{scope}
  \end{scope}

  % (c) Linear image, preserving the lengths of the grid sides.
  \begin{scope}[shift={(8,0)}]
    \begin{scope}[shift={(0.40,0.40)}]
      \begin{scope}[
        x={(0.55cm,0cm)},
        y={(0.275cm,0.476314cm)}
      ]

        \foreach \i in {0,...,4} {
          \draw[embgrid] (\i,0)--(\i,4);
          \draw[embgrid] (0,\i)--(4,\i);
        }

        \draw[embedge] (0,2)--(0,4)--(2,4);
        \draw[embedge] (2,4)--(4,4)--(4,1);
        \draw[embedge] (0,2)--(2,2)--(2,1)--(4,1);
        \draw[embedge] (0,2)--(0,0)--(2,0);
        \draw[embedge] (2,0)--(4,0)--(4,1);

        \foreach \x/\y in {0/4,4/4,0/0,4/0,2/2,2/1}
          \fill[embbend] (\x,\y) circle[radius=0.045cm];

        \foreach \x/\y in {0/2,4/1,2/4,2/0}
          \fill (\x,\y) circle[radius=0.06cm];

        \node[anchor=east,xshift=-4pt] at (0,2) {$u$};
        \node[anchor=west,xshift=4pt] at (4,1) {$v$};
        \node[above,yshift=3pt] at (2,4) {$p$};
        \node[below,yshift=-3pt] at (2,0) {$q$};

      \end{scope}
    \end{scope}
  \end{scope}

  % (d) Triangular lattice, at the same scale as panel (c).
  \begin{scope}[shift={(12,0)}]
    \begin{scope}[shift={(0.40,0.40)}]
      \begin{scope}[
        x={(0.55cm,0cm)},
        y={(0.275cm,0.476314cm)}
      ]

        % The two transformed grid-line families.
        \foreach \i in {0,...,4} {
          \draw[embgrid] (\i,0)--(\i,4);
          \draw[embgrid] (0,\i)--(4,\i);
        }

        % The third family completes the triangular lattice.
        \foreach \i in {0,...,3} {
          \foreach \j in {0,...,3} {
            \draw[embgrid] ({\i+1},\j)--(\i,{\j+1});
          }
        }

        \draw[embedge] (0,2)--(0,4)--(2,4);
        \draw[embedge] (2,4)--(4,4)--(4,1);
        \draw[embedge] (0,2)--(2,2)--(2,1)--(4,1);
        \draw[embedge] (0,2)--(0,0)--(2,0);
        \draw[embedge] (2,0)--(4,0)--(4,1);

        \foreach \x/\y in {0/4,4/4,0/0,4/0,2/2,2/1}
          \fill[embbend] (\x,\y) circle[radius=0.045cm];

        \foreach \x/\y in {0/2,4/1,2/4,2/0}
          \fill (\x,\y) circle[radius=0.06cm];

        \node[anchor=east,xshift=-4pt] at (0,2) {$u$};
        \node[anchor=west,xshift=4pt] at (4,1) {$v$};
        \node[above,yshift=3pt] at (2,4) {$p$};
        \node[below,yshift=-3pt] at (2,0) {$q$};

      \end{scope}
    \end{scope}
  \end{scope}

\end{tikzpicture}
    \caption{From $\mathcal{G}$ to a triangular-lattice drawing: the graph, an
    orthogonal drawing, its image under the linear map $T$, with both grid
    directions transformed, and the same drawing with the third lattice
    direction added. Black dots are vertices of $\mathcal{G}$; orange dots mark bends.}
    \label{fig:embedding-stages}
\end{figure}
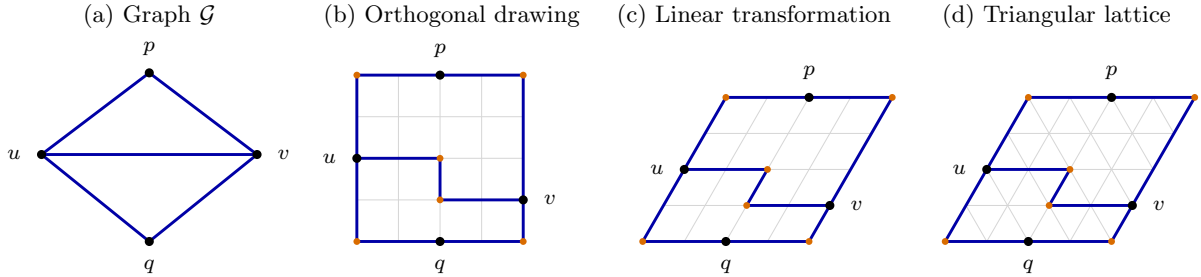

The \emph{AND and OR cores} are the constant-size central parts of the
vertex gadgets. Their terminal contacts enforce the inflow constraints
through the One Hive rule (Section~\ref{sec:vertex-gadgets}).
Fix a constant-size region large enough for a core and its ports,
or for one of the required straight and turn patterns.
Reserve a region of this size at every mapped grid point within
the drawing's bounding box. A source vertex receives its core;
an intermediate point of an arc receives a transit pattern.
Points off the drawing need no pieces. These are reservations
of space, not instructions to fill every cell.
The margin is chosen after fixing the core, port, and turn patterns.
It preserves their local contacts while separating unrelated structures.

We now scale the drawing by $s+O(1)$, choosing the constant margin so that each unit step of the mapped grid has $s$ consecutive Hive cells between the regions at its ends. The AND and OR cores therefore have room at graph vertices, while every straight run of an arc contains at least $s$ usable cells. Figure~\ref{fig:embedding-buffers} illustrates one unit step.

\begin{figure}[H]
    \centering
    \begin{tikzpicture}[font=\small]
      \begin{scope}[shift={(2.4,2.2)}]
        \begin{scope}[x={(0.42cm,0cm)},y={(0.21cm,0.364cm)}]
          \tikzset{embedding hex/.style={regular polygon,
            regular polygon sides=6,shape border rotate=30,
            minimum size=0.47cm,inner sep=0pt,line width=0.5pt}}
          \def\hexcell#1#2#3{%
            \node[embedding hex,draw=white,fill=#3] at (#1,#2) {};}
          \def\whitehexcell#1#2{%
            \node[embedding hex,draw=hivewhiteoutline,fill=hivewhiteedge]
              at (#1,#2) {};}
          % Equal-sized regions at two consecutive coarse-grid points.
          \foreach \i in {-2,-1,0,1,2} {
            \foreach \j in {-2,-1,0,1,2} {
              \pgfmathtruncatemacro{\dist}
                {max(abs(\i),max(abs(\j),abs(\i+\j)))}
              \ifnum\dist<3 \hexcell{\i}{\j}{hivepatch} \fi
            }
          }
          \foreach \di in {-2,-1,0,1,2} {
            \foreach \j in {-2,-1,0,1,2} {
              \pgfmathtruncatemacro{\dist}
                {max(abs(\di),max(abs(\j),abs(\di+\j)))}
              \pgfmathtruncatemacro{\ii}{16+\di}
              \ifnum\dist<3 \hexcell{\ii}{\j}{hivepatch} \fi
            }
          }
          % Usable edge-row cells and the other two ports.
          \foreach \i in {3,...,13} {
            \hexcell{\i}{0}{hiveedge}
          }
          \foreach \j in {3,4,5} {
            \whitehexcell{0}{\j}
            \whitehexcell{0}{-\j}
          }
          \foreach \i/\j in {2/0,0/2,0/-2} {
            \fill[black!70] (\i,\j) circle[radius=0.055cm];
          }
          \fill[black!80] (0,0) circle[radius=0.09cm];
          \fill[black!80] (16,0) circle[radius=0.09cm];
          \node[font=\tiny,text=white] at (0,0) {$v$};
        \end{scope}
      \end{scope}
      \node[anchor=west] at (0.1,4.25) {graph-vertex core};
      \draw[-{Latex[length=2mm]},thin] (1.9,4.10)--(2.42,2.37);
      \node[align=center] at (11.1,4.12) {next coarse\\grid point};
      \draw[-{Latex[length=2mm]},thin] (10.65,3.90)--(9.14,2.36);
      \draw[gray!70!black] (3.61,1.28)--(7.91,1.28);
      \draw[gray!70!black] (3.61,1.19)--(3.61,1.37);
      \draw[gray!70!black] (7.91,1.19)--(7.91,1.37);
      \node[gray!70!black] at (5.76,0.94) {$s$ usable cells};
      
    \end{tikzpicture}
    \caption{One stretched unit step. Purple patches are the same constant-size
    regions reserved at every coarse-grid point. Dark-gray and off-white cells
    distinguish edge arcs.
    A patch may contain a graph-vertex core, a transit
    pattern, or no pieces at all.}
    \label{fig:embedding-buffers}
\end{figure}
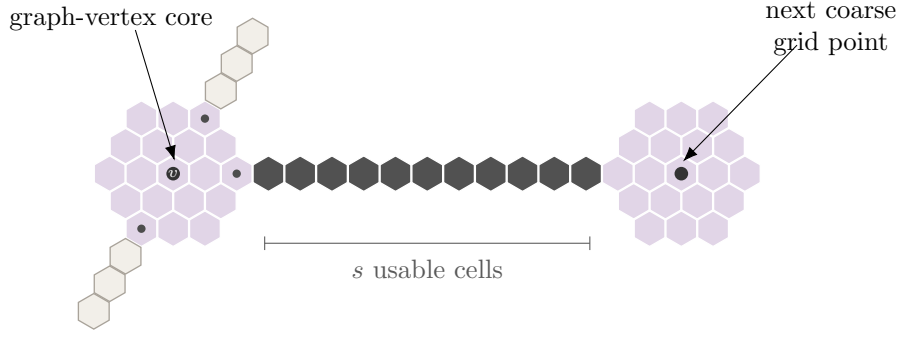

\begin{lemma}
\label{lem:layout}
For every graph $\mathcal{G}$ in the source family and every sufficiently large integer
$s$, the following layout can be constructed in time polynomial in $N+s$:
\begin{enumerate}[label=\textup{(\arabic*)}]
\item $\mathcal{G}$ has a planar triangular-lattice drawing with at most two bends per
      edge and total edge length $O(N^2s)$; each straight segment
      contains at least $s$ usable Hive cells;
\item every point of the mapped grid has a constant-size reserved region, and
      fixed-width strips belonging to nonincident arcs can be kept disjoint;
\item every bounded face of the drawing contains a parallelogram of side
      $\Theta(s)$, hence $\Omega(s^2)$ Hive-cell centres.
\end{enumerate}
\end{lemma}

\begin{proof}
The orthogonal drawing has total edge length $O(N^2)$. The linear map and stretch
preserve its embedding and bends, while multiplying its lengths by $s+O(1)$.
Distinct nonincident features have positive, constant separation in the
integer-grid drawing; after stretching, their fixed-width regions have the
required clearance.

The mapped drawing uses only the $a$- and $b$-directions of the triangular
lattice; no graph-edge segment uses the third direction, $b-a$. Its edges
therefore lie on the boundaries of the elementary parallelograms of the
mapped square grid and cannot cross their interiors. Every bounded face
consequently contains the interior of at least one such parallelogram,
the image under $T$ of an open unit grid square. After stretching, its sides
have length $s+O(1)$. Leaving a constant-width margin inside it still gives
$\Omega(s^2)$ Hive-cell centres.
\end{proof}

\subsection{Edge gadgets and covering stacks}
\label{sec:edge-gadget}

Fix an edge $e=uv$ owned by $C\in\{W,B\}$, and write $\overline{C}$
for the other player. We represent $e$ by one Grasshopper $G^C$
and a straight occupied row with terminals $t_u$ and $t_v$.
The Grasshopper occupies exactly one terminal: its presence at
$t_v$ means that $e$ is directed into $v$.

The row cell adjacent to each terminal contains a
\emph{covering stack} controlled by $\overline{C}$, the opponent of the
Grasshopper's owner. Fix a constant $h\ge4$ and write
\[
  MS^{\overline{C}}=(\mathrm{Bt}^{C},
        M^{C}_1,\ldots,M^{C}_h,M^{\overline{C}}),
\]
from bottom to top. The top Mosquito is the \emph{lid};
the intermediate Mosquitoes form the \emph{body}.
The superscript denotes the lid's owner: the lid belongs to $\overline{C}$,
whereas all body Mosquitoes and the bottom Beetle belong to $C$.
Thus moving the lid exposes a piece belonging to the Grasshopper's owner.
Each stack has height $h+2$. Four body Mosquitoes suffice: one remains
above the bottom Beetle to prevent other Mosquitoes from copying its
movement, two become successive rovers, and one countercovers the
displaced lid. The choice $h=4$ therefore gives a six-piece stack,
ordered from bottom to top as Beetle, copy-blocking Mosquito, second
rover, first rover, countercover, and lid.
The covering stacks defined here occur only next to edge-gadget terminals;
the distinct Beetle stacks in the royal regions are described in
Section~\ref{sec:royal}. Each terminal
is adjacent to the stack on its edge row and, as described in
Section~\ref{sec:vertex-gadgets}, to one or two stacks in the incident
vertex gadget. Thus ordinary terminals have two adjacent stacks, whereas
the weight-$2$ terminal of an AND gadget has three.
All stacks adjacent to a terminal of $G^C$ have this same ownership:
their lids are controlled by $\overline{C}$, and their bodies belong to $C$.
The presence of the Beetles at the bottoms of the stacks also makes
this configuration constructible by legal moves: a ground-level
Mosquito can climb only by copying a Beetle.
Without any Beetles in the construction, no initially flat Mosquito
could ever rise above ground level.

We now specify the edge row. Choose $\ell\ge6$ consecutive cells
$k_0,\ldots,k_{\ell-1}$ on a straight lattice line.
The terminals are $t_v=k_0$ and $t_u=k_{\ell-1}$.
Place a distinct copy of $MS^{\overline{C}}$ at each of
$k_1$ and $k_{\ell-2}$, and a flat Mosquito of colour $C$ at every cell
$k_2,\ldots,k_{\ell-3}$.
This colour choice visually identifies the row as an edge owned by $C$;
it is not required for the operation of the gadget.
This middle segment is the \emph{dead zone} $D_e$.
During faithful play, $G^C$ occupies one terminal and the other
terminal is empty.

From either occupied terminal, the Grasshopper rule selects the
opposite terminal as the landing cell: every intervening cell
is occupied and that terminal is empty.
The jump is legal if and only if the One Hive rule permits its
departure. This is the key to the simulation: under faithful play,
the vertex gadget permits that departure if and only if the
corresponding outward reversal satisfies the inflow constraint.
Thus the row determines where the Grasshopper lands,
while the vertex determines whether it may leave.
\vspace{-1.5em}

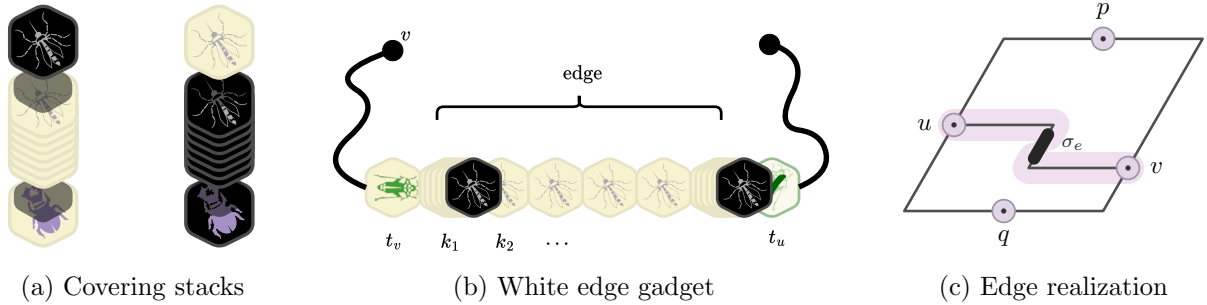
\begin{figure}[H]
\centering
% Vertical adjustment for panel (b): a more negative value moves only
% the edge-gadget image farther down, without moving its subcaption.
\def\edgegadgetyshift{-0.50cm}
\begin{minipage}[b]{0.22\linewidth}
  \centering
  \includegraphics[width=\linewidth,height=3.4cm,keepaspectratio]
    {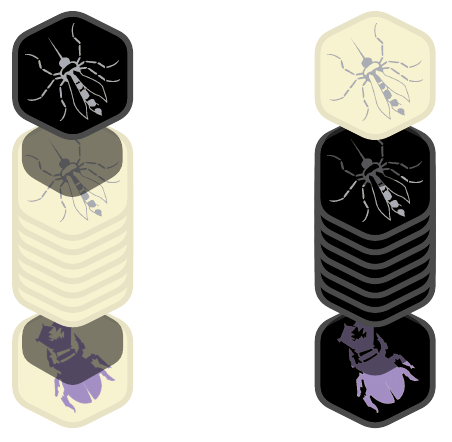}
  \par\smallskip\small (a) Covering stacks
\end{minipage}\hfill
\begin{minipage}[b]{0.45\linewidth}
  \centering
  \raisebox{\edgegadgetyshift}[\height][\depth]{%
    \includegraphics[width=\linewidth]{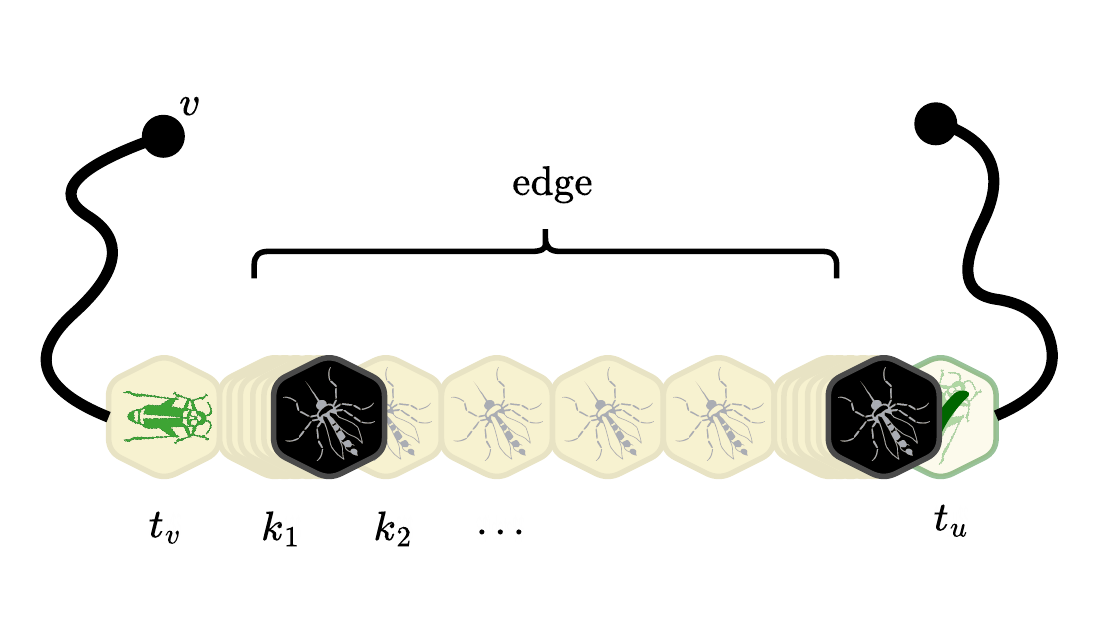}%
  }
  \par\smallskip\small (b) White edge gadget
\end{minipage}\hfill
\begin{minipage}[b]{0.25\linewidth}
  \centering
  \resizebox{\linewidth}{!}{%
  \begin{tikzpicture}[
    font=\small,
    graphedge/.style={
      draw=gray!62!black,line width=1.05pt,
      line cap=round,line join=round},
    armhalo/.style={
      draw=violet!55,opacity=0.23,line width=0.45cm,
      line cap=round,line join=round}
  ]
    \begin{scope}[x={(0.70cm,0cm)},y={(0.35cm,0.606cm)}]
      % The vertex structures follow the two ends of the polygonal edge.
      \draw[armhalo] (0,2)--(2,2)--(2,1.78);
      \draw[armhalo] (4,1)--(2,1)--(2,1.22);

      \draw[graphedge] (0,2)--(0,4)--(2,4);
      \draw[graphedge] (2,4)--(4,4)--(4,1);
      \draw[graphedge] (0,2)--(2,2)--(2,1)--(4,1);
      \draw[graphedge] (0,2)--(0,0)--(2,0);
      \draw[graphedge] (2,0)--(4,0)--(4,1);

      \draw[gray!30!black,line width=5pt,line cap=round]
        (2,1.78)--(2,1.22);
      \fill[gray!30!black] (2,1.78) circle[radius=0.055cm];
      \fill[gray!30!black] (2,1.22) circle[radius=0.055cm];

      \foreach \x/\y in {0/2,4/1,2/4,2/0} {
        \filldraw[fill=hivepatch,draw=hivepatch!70!black,line width=0.7pt]
          (\x,\y) circle[radius=0.16cm];
        \fill[black!80] (\x,\y) circle[radius=0.037cm];
      }
      \node[anchor=east,xshift=-5pt] at (0,2) {$u$};
      \node[anchor=west,xshift=5pt] at (4,1) {$v$};
      \node[above,yshift=4pt] at (2,4) {$p$};
      \node[below,yshift=-4pt] at (2,0) {$q$};
      \node[font=\scriptsize,text=gray!30!black,anchor=west]
        at (2.2,1.55) {$\sigma_e$};
    \end{scope}
  \end{tikzpicture}%
  }
  \par\smallskip\small (c) Edge realization
\end{minipage}
\caption{The edge construction. (a) The two colour-swapped covering
stacks. (b) The Grasshopper at $t_v$ encodes orientation into $v$.
(c) Vertex structures extend along a polygonal edge to the straight
segment $\sigma_e$ carrying its edge gadget.}
\label{fig:stacktrap1}
\label{fig:edge-hive}
\label{fig:polygonal-arc}
\end{figure}

\paragraph{Polygonal edges.}
\label{sec:polygonal}
The Grasshopper jumps along a straight occupied line and cannot
follow the bends of a polygonal edge. For each ordinary edge $e=uv$, choose one
straight run of its scaled arc, possibly between two bends. It
contains at least $s$ usable cells; for sufficiently large $s$,
reserve a subsegment $\sigma_e$ for the unique edge gadget and its
terminal interfaces. The portions of the arc between $\sigma_e$ and
the two vertices $u$ and $v$ become stationary arms of their respective
vertex structures; either portion may be empty. The vertex cores
remain at the grid points representing $u$ and $v$, while their
structures extend along the drawn arc to the edge gadget.
An arm extends the same floating components as its vertex core.
In particular, the two components reaching a weight-$2$ AND terminal
follow separate tracks and meet only through its Grasshopper.
Straight and turn patterns must preserve these separations and avoid
contacts with the fixed edge row outside the terminal interface.
A turn bends the separated tracks together without joining them;
the constant-width space reserved for the arm accommodates both tracks.

At the ends of $\sigma_e$, the arms use the prescribed
terminal interfaces, leaving the gadget's straight Grasshopper run
unchanged and the cell immediately beyond each terminal empty. We have only
partitioned the original arc into one edge gadget and two possible
arms, so their total length remains $O(N^2s)$.
For a target edge, the winning-side arm is omitted as described in
Section~\ref{sec:royal}.

Every occupied neighbour of a dead-zone Mosquito exposes another
Mosquito, including the lids and the flat Mosquitoes at connector
attachments. Hence no dead-zone piece has a movement ability to
copy during faithful play. Covered stack pieces cannot act either.

The lids themselves can move as Beetles. The responses to their
moves are analyzed in Section~\ref{sec:deviations-traps}, together
with the Grasshopper traps.
\subsection{Vertex gadgets}
\label{sec:vertex-gadgets}

A vertex gadget determines whether an incident Grasshopper may
leave its terminal. It is formed from Mosquitoes arranged into
connected groups whose only contacts with the incident edge gadgets pass
through occupied terminals. We call these groups \emph{floating
components}. The purple structural Mosquitoes in
Figure~\ref{fig:vertex-gadgets} may each be of either colour, chosen
independently: their colour affects neither the connectivity tests nor
the immobility arguments below.
At each interface with an incident edge gadget, the cell
of every floating component adjacent to the terminal contains a distinct
copy of the covering stack of Section~\ref{sec:edge-gadget}.
For an incident Grasshopper $G^C$, this is $MS^{\overline{C}}$. Together with the edge-row stack, this gives two
stacks at an ordinary terminal. The
weight-$2$ AND terminal is adjacent to both floating components and hence
has three. Thus the opponent of the Grasshopper's owner controls all
stacks at that terminal.

Let $t_1,t_2,t_3$ be the incident terminals. In an AND gadget,
$t_3$ represents the weight-$2$ edge and there are two floating
components: $K_1$ may contact the exterior through $t_1$ or $t_3$,
whereas $K_2$ may do so through $t_2$ or $t_3$. An OR gadget has one
floating component $K$, which may contact the exterior through any of
$t_1,t_2,t_3$. Only occupied terminals supply these contacts, as shown
in Figure~\ref{fig:vertex-gadgets}. The edge rows and face connectors
together form the connected network constructed in
Section~\ref{sec:onehive}.

\begin{figure}[H]
    \centering
    \begin{tikzpicture}
      \node[anchor=south west,inner sep=0] (vertexpicture) at (0,0)
        {\includegraphics[width=1.0\linewidth,trim=0 70 0 0,clip]{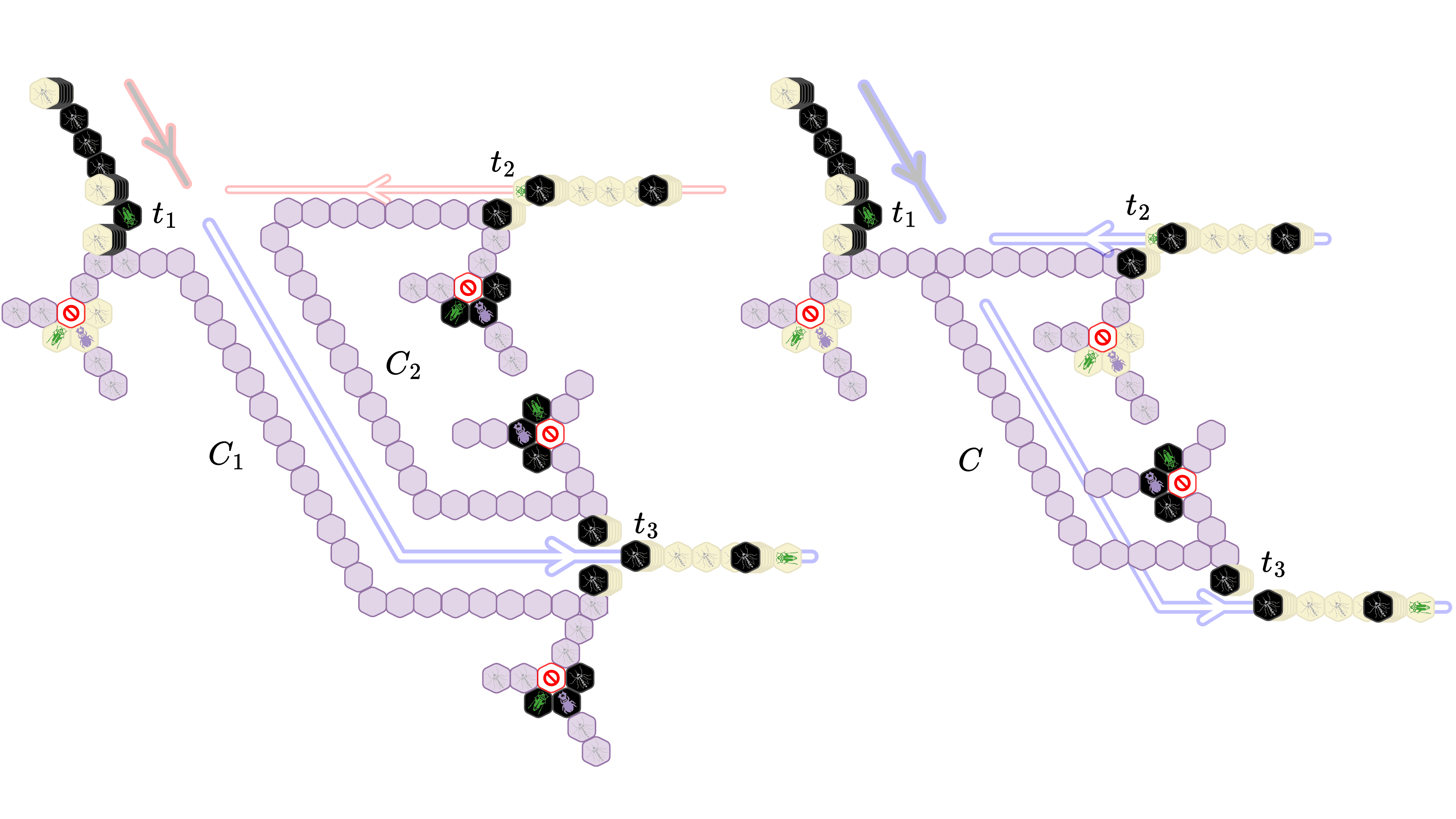}};
      % Relabel the floating components in the original Keynote figure.
      \begin{scope}[x={(vertexpicture.south east)},y={(vertexpicture.north west)}]
        \node[fill=white,inner sep=1pt,font=\large] at (0.155,0.405) {$K_1$};
        \node[fill=white,inner sep=1pt,font=\large] at (0.275,0.524) {$K_2$};
        \node[fill=white,inner sep=1pt,font=\large] at (0.667,0.401) {$K$};
      \end{scope}
    \end{tikzpicture}
    \caption{Left: the AND gadget. Right: the OR gadget. Purple hexagons represent Mosquitoes. The red prohibition
    symbols mark empty cells that the encoding Grasshoppers should not
    enter. The pieces beside these cells
    form the trap mechanisms described in
    Section~\ref{sec:deviations-traps}.}
    \label{fig:vertex-gadgets}
\end{figure}
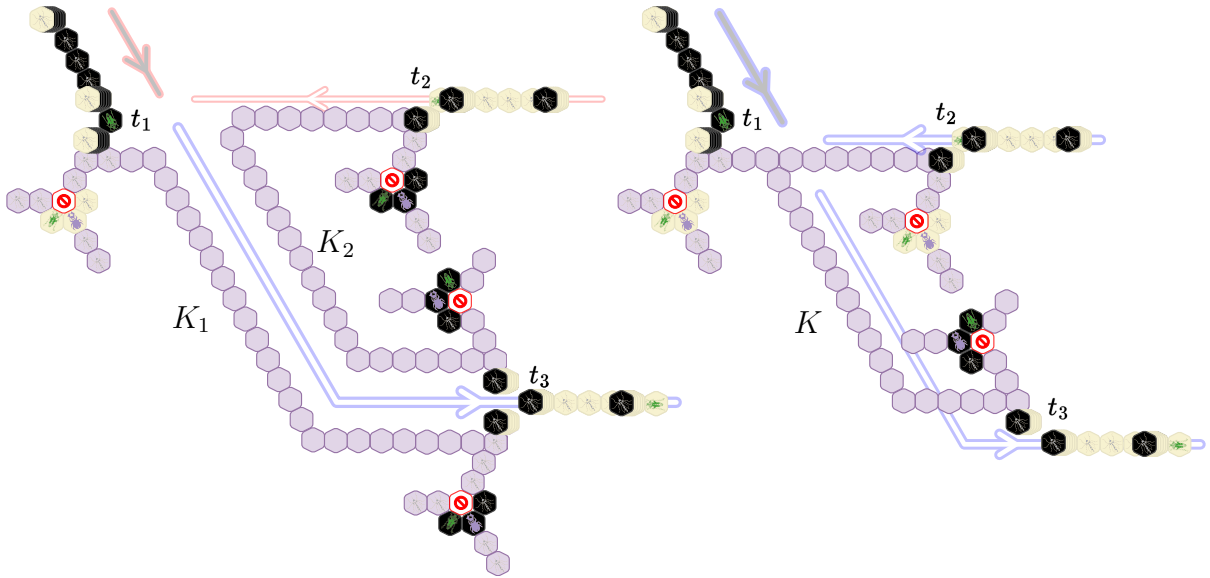

The remaining vertex-gadget pieces are arranged to be immobile during
faithful play. Each uncovered internal Mosquito is either adjacent only to
Mosquitoes, and hence copies no movement ability, or belongs to a
paired-Mosquito pin. In this pattern, described in the preliminaries,
One Hive pins the auxiliary piece and the inner Mosquito, while the
terminal Mosquito has no ability to copy. The same local pattern is used by the trap supports in
Section~\ref{sec:deviations-traps}. All remaining stack pieces are covered.

\begin{lemma}[Vertex gadgets]
\label{lem:vertex-gating}
Assume that the incident edge gadgets are connected by the face connectors.
In a faithful configuration:
\begin{enumerate}[label=\textup{(\roman*)}]
\item all floating components are attached to this network if and only if
      the source inflow constraint is satisfied;
\item the designated terminal-to-terminal jump of an incident Grasshopper
      is legal if and only if the corresponding edge reversal is legal;
\item apart from incident Grasshoppers and exposed covering-stack lids,
      all vertex-gadget pieces are immobile.
\end{enumerate}
\end{lemma}

\begin{proof}
Let $x_i$ indicate whether $t_i$ is occupied. For an AND vertex,
both components are attached precisely when
\[
  (x_1\lor x_3)\land(x_2\lor x_3),
\]
equivalently $x_1+x_2+2x_3\ge2$; for an OR vertex, attachment is
equivalent to $x_1\lor x_2\lor x_3$, or
$2(x_1+x_2+x_3)\ge2$. These are exactly the two source constraints, proving (i).
Changing one occupied $x_i$ from $1$ to $0$ models the departure of its
Grasshopper. Arrival at the other endpoint increases that vertex's inflow,
so only the departure vertex can make the edge reversal illegal.
Thus the designated terminal-to-terminal jump is legal if and only if
the corresponding edge reversal is legal, proving (ii).
The Mosquito-only adjacencies, paired pins, and covered stack pieces
described above prove (iii).
\end{proof}

The construction in Section~\ref{sec:onehive}
connects the edge gadgets into a single network
without creating additional contacts with the floating components.
The branches in the drawings realize these contacts:
the two AND components remain disjoint until they meet an occupied
terminal, whereas the OR component joins its three terminal branches.
The vertex arms extend these same components, and the face connectors
attach to the edge rows away from the terminals.
\subsection{Deviations and traps}
\label{sec:deviations-traps}

Consider a faithful configuration of the edge and vertex gadgets
constructed so far. The only pieces that can have a legal move are the
exposed lids of covering stacks and the Grasshoppers.
Auxiliary pieces introduced later are either immobile during faithful play
or have their possible moves analyzed separately, as for the pass gadgets
and royal regions (Sections~\ref{sec:pass} and~\ref{sec:royal}).

Not every available move of these two classes is a faithful move. A move of a
covering-stack lid does not represent any source move. A Grasshopper move is
faithful only when it is the designated jump representing a legal source-edge
reversal or source pass; the same Grasshopper may also be able to jump along
an unintended occupied ray. Thus the lid moves and the unintended
Grasshopper jumps are precisely the deviations analyzed below.

Both kinds of deviation have a prescribed answer that loses no tempo.
The opponent's next move neutralizes the deviation and starts releasing
opposing Mosquitoes above the hive; we call such a Mosquito a
\emph{rover}. A lid deviation makes two rovers available in succession:
when the first leaves the source stack, it exposes the second. Only one of
them must reach an edge gadget. After a Grasshopper deviation, the rover's
first route move is the climb described below.

\paragraph{Covering-stack traps.}
\label{sec:stack-trap}
Let $t$ be a terminal of $G^C$, guarded by copies of $MS^{\overline{C}}$.
If $\overline{C}$ moves a lid from its stack at $x$ to an adjacent cell $y$,
it exposes the Mosquito $M_h^C$ above ground level. Player $C$
immediately covers the lid with this Mosquito in one Beetle move and
leaves it there as a permanent guard. This prevents the displaced lid
from copying a Grasshopper or moving again. The same countercover
applies whether the lid descends to the ground or moves above the hive,
and whether $t$ is empty or occupied; Figure~\ref{fig:stacktrap-empty}
shows examples at OR and AND terminals.
The number and spacing of the elevated stacks at terminals and
Grasshopper traps ensure that no height gate obstructs the prescribed
countercovers or the escape of the released Mosquitoes; this property
requires a more detailed analysis for the Beetle construction
in Section~\ref{sec:beetle-stacks}.\\
The reply leaves the source cell occupied and exposes the next body Mosquito
as a rover. When this rover leaves, a second becomes available.
Since $h\ge4$, at least one body Mosquito remains above the bottom
Beetle after the countercover, preventing other Mosquitoes from copying the bottom Beetle's movement ability. Several lids
may deviate on different turns, but each
receives its own immediate countercover, whose guard never transfers.

\begin{figure}[H]
    \centering
    \begin{minipage}[t]{0.175\linewidth}
        \centering
        \raisebox{0.50cm}[\height][\depth]{%
            \includegraphics[width=\linewidth]{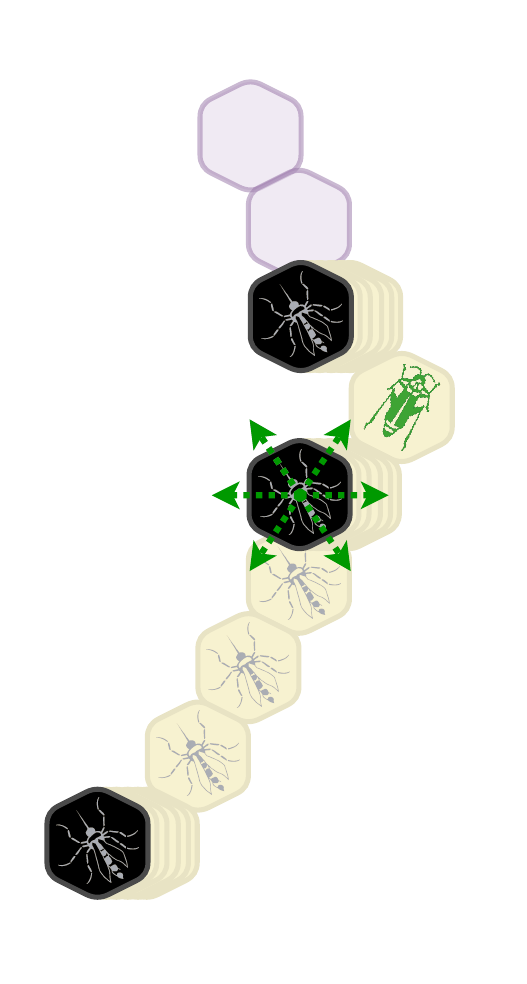}%
        }
        \par\smallskip\scriptsize (a) The deviation at the OR terminal.
    \end{minipage}\hfill
    \begin{minipage}[t]{0.19\linewidth}
        \centering
        \raisebox{0.35cm}[\height][\depth]{%
            \includegraphics[width=\linewidth]{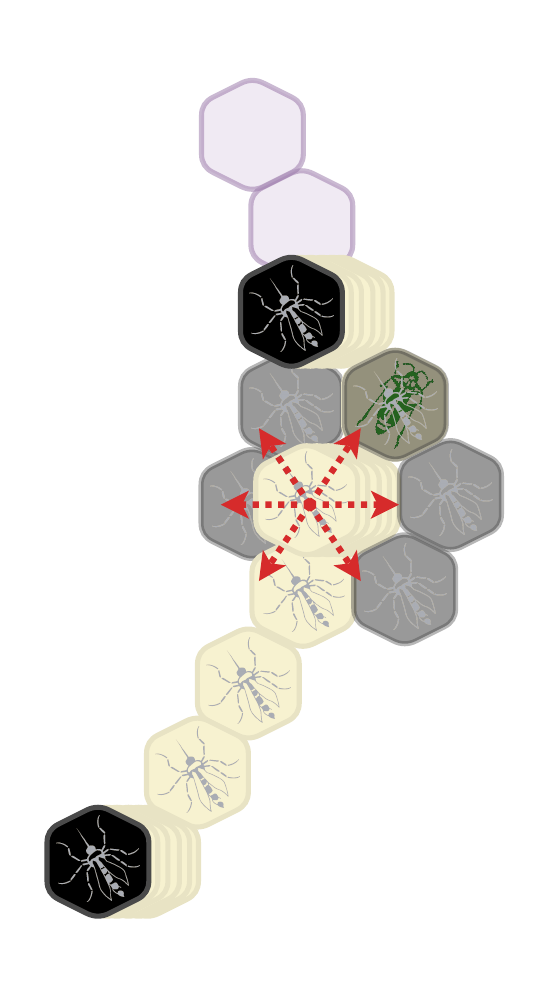}%
        }
        \par\smallskip\scriptsize (b) The exposed Mosquito covers it.
    \end{minipage}\hfill
    \begin{minipage}[t]{0.29\linewidth}
        \centering
        \includegraphics[width=\linewidth]{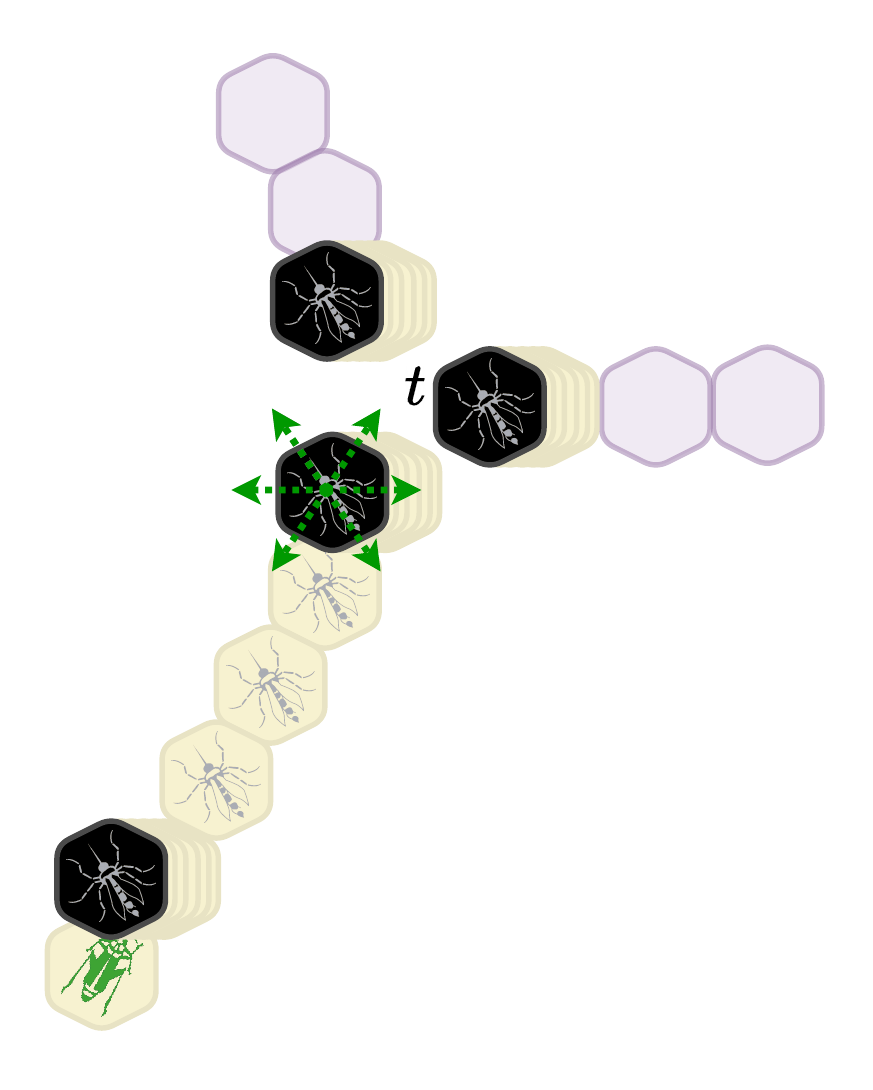}
        \par\smallskip\scriptsize (c) The deviation at the AND terminal.
    \end{minipage}\hfill
    \begin{minipage}[t]{0.29\linewidth}
        \centering
        \includegraphics[width=\linewidth]{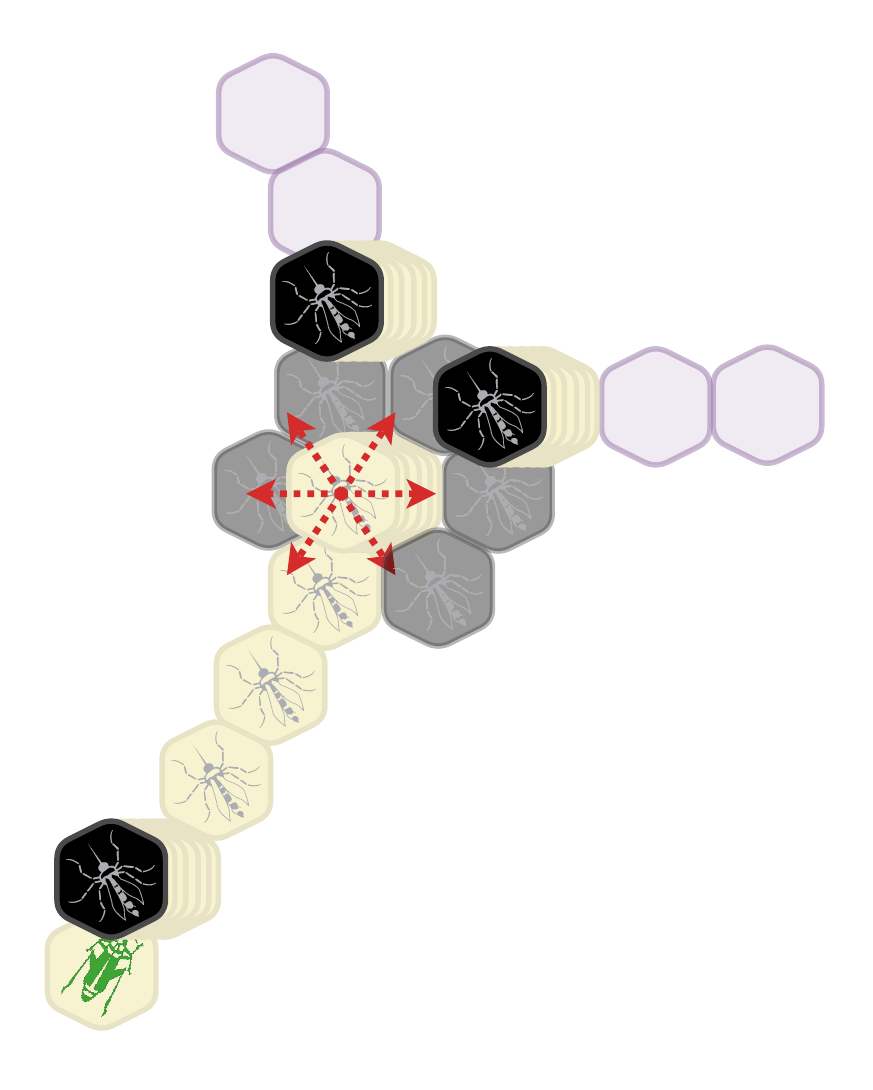}
        \par\smallskip\scriptsize (d) The exposed Mosquito covers it.
    \end{minipage}
    \caption{The same countercover at an empty terminal of an OR
    (a)--(b) and an AND (c)--(d) vertex.}
    \label{fig:stacktrap-empty}
\end{figure}

Ground-level lid moves and their countercovers may change occupied rays
or contacts near $t$: $G^C$ may acquire new landings or become blocked,
and the source inflow constraints may be relaxed. These changes need
not be repaired. The affected Grasshopper belongs to $C$, who can leave
it fixed and pursue punishment with the released rovers instead of
continuing the source simulation. Even if the relaxed constraints allow
$\overline{C}$ additional edge reversals, they do not shorten its protected
royal completion. The global timing argument in
Section~\ref{sec:royal-punishment} shows that rover punishment finishes
first. Thus no Grasshopper interception is needed.

\paragraph{The Grasshopper trap.}
\label{sec:grasshopper-trap}

The pieces of a vertex gadget may give an encoding
Grasshopper additional jump directions from its terminal, besides the
intended edge row. A jump along one of these rays does not
represent an edge reversal; we call it a \emph{Grasshopper deviation}.

For each such ray, we place a trap inside the vertex gadget at constant
cell-scale distance from the terminal of departure $t$. Let $z$ be the
first empty cell along the ray. Figure~\ref{fig:trap-hive} shows the local
arrangement and sequence. The pieces involved are the encoding Grasshopper
$G^C$, an opposing Grasshopper $G^{\overline C}$, an opposing Mosquito
$M^{\overline C}$, and an opposing Beetle $\mathrm{Bt}^{\overline C}$ adjacent to
$M^{\overline C}$. For readability, throughout this paragraph we abbreviate them
as $G$, $\overline G$, $\overline M$, and $\overline{\mathrm{Bt}}$, respectively.
The trap support attaches to the corresponding vertex component.
Before the deviation, $\overline G$ and $\overline M$
are pinned by the One Hive rule. The adjacent $\overline{\mathrm{Bt}}$ is
also a cut piece and supplies the movement ability later copied by $\overline M$.

When $G$ lands at $z$, it supplies an alternative connection
that releases $\overline G$. This Grasshopper is positioned to
jump across the occupied ray to the vacated terminal $t$, the
first empty cell in that direction.
The reply restores the previous pattern of occupied terminals.
The full position is no longer faithful, since $G$ remains at $z$
and the piece at $t$ belongs to $\overline C$.

Removing $\overline G$ from its original cell makes $G$ at $z$
the sole connection of the trap's pendant support. Thus after the
terminal replacement, $G$ is immobilized and $\overline M$ is free to act as a rover.

The released Mosquito can begin its punishment route by copying
$\overline{\mathrm{Bt}}$ and climbing directly onto the occupied support cell
immediately beyond $G$ in the direction of $t$. Removing
$\overline M$ leaves the support connected through $z$, and the two common
neighbours of this Beetle step do not both reach travel level two, so they
do not block the climb. This move is already the first step of the rover's
route, not a second trapping response.

Intervening turns cannot remove this opportunity. The trapped
Grasshopper $G$ and the adjacent $\overline{\mathrm{Bt}}$ are cut pieces;
the local support remains
in place, and no external move reaches the cells needed for the climb.
A deviation elsewhere is neutralized at its own guard; it can delay
the climb while that response is made, but cannot disable it.

\begin{figure}[H]
  \centering
  \begin{minipage}[t]{0.31\linewidth}
    \centering
    \includegraphics[width=\linewidth]{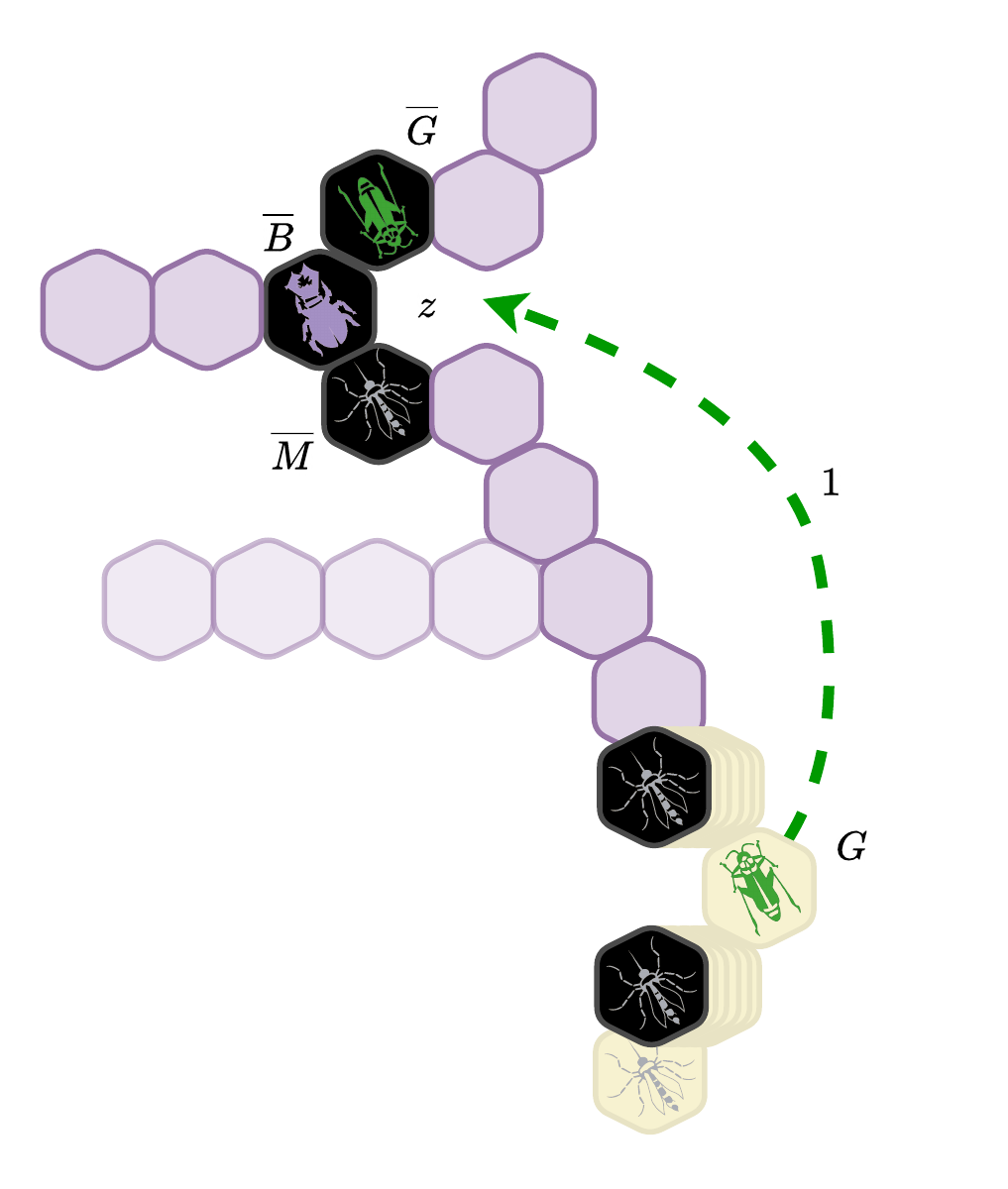}
    \par\smallskip\scriptsize (a) $G$ deviates and lands at $z$.
  \end{minipage}\hfill
  \begin{minipage}[t]{0.31\linewidth}
    \centering
    \includegraphics[width=0.87\linewidth]{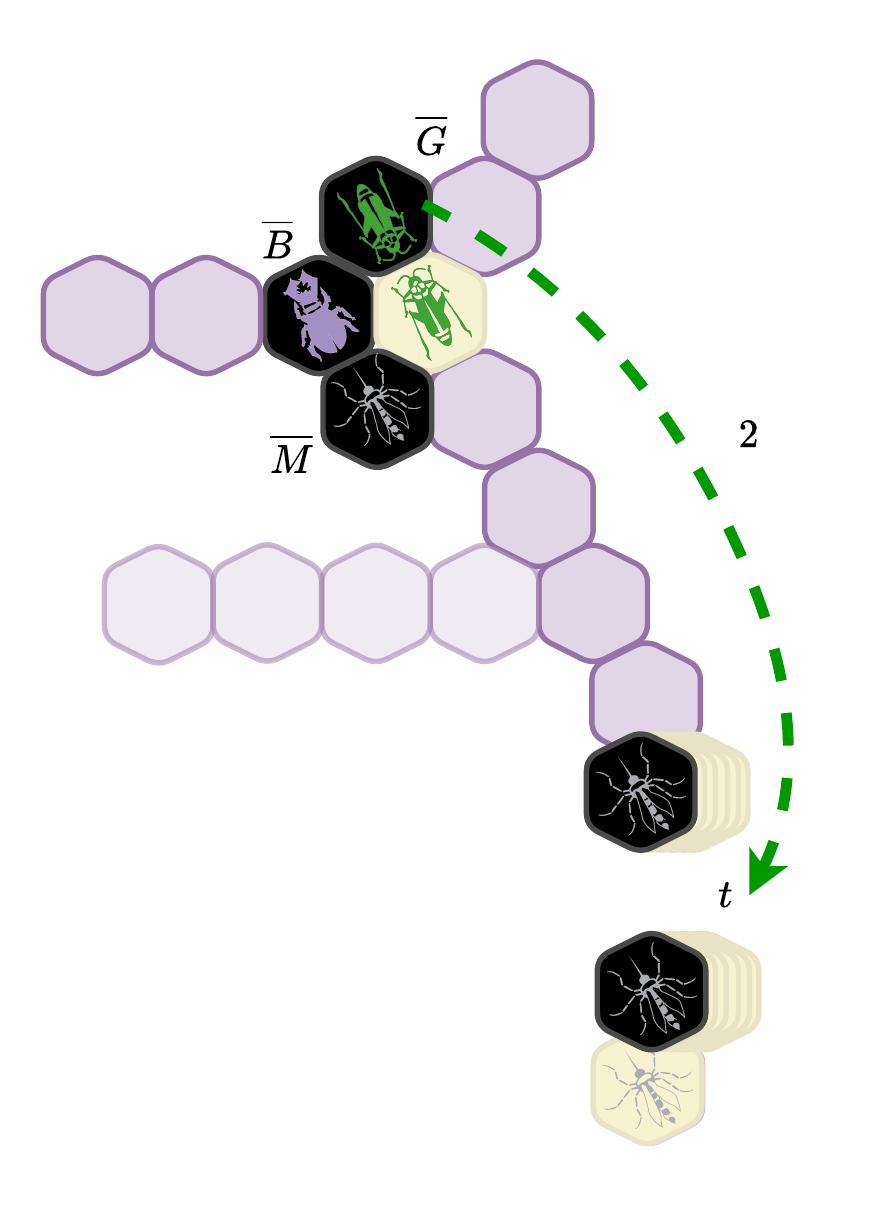}
    \par\smallskip\scriptsize (b) $\overline G$ restores terminal $t$.
  \end{minipage}\hfill
  \begin{minipage}[t]{0.31\linewidth}
    \centering
    \includegraphics[width=0.97\linewidth]{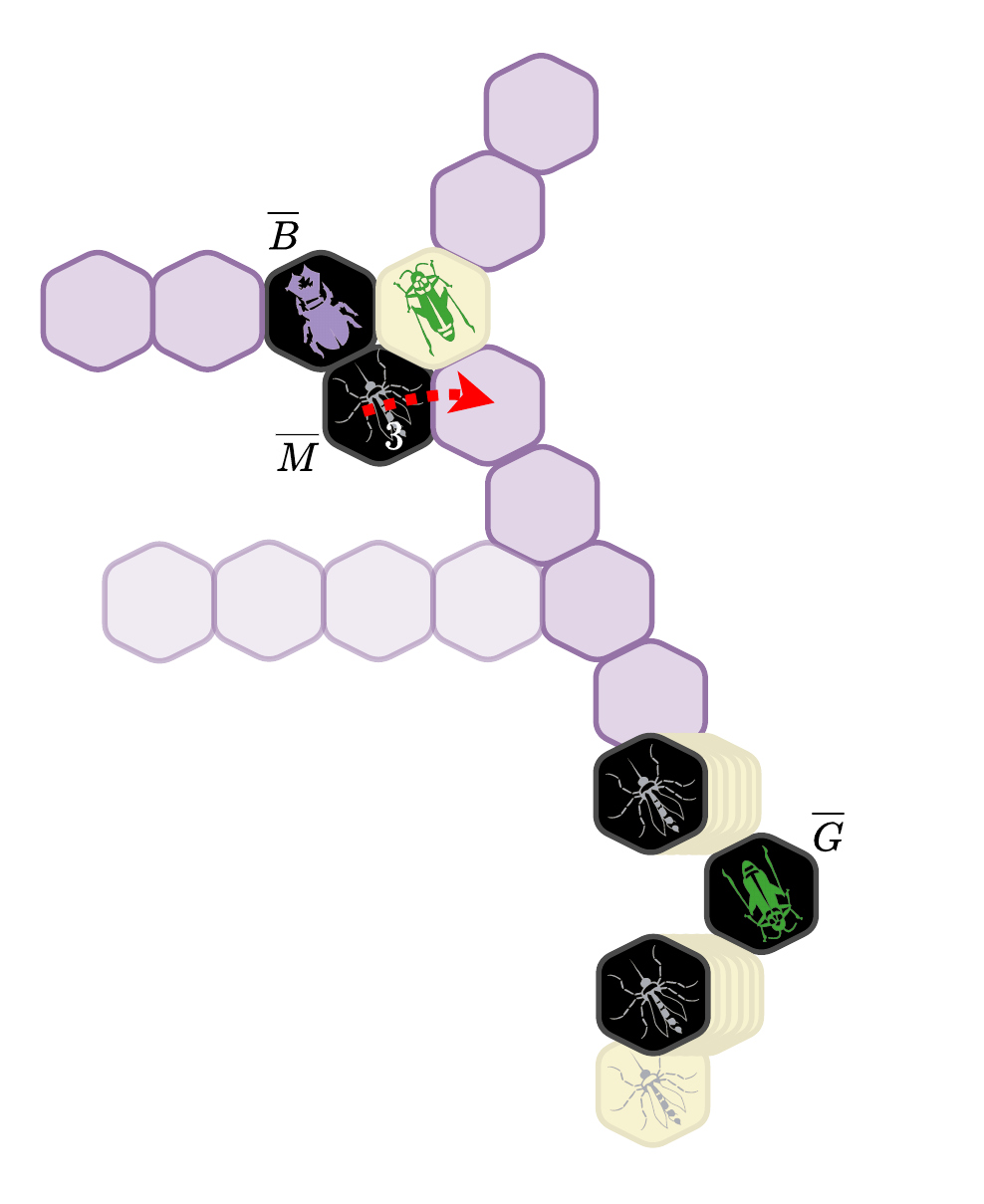}
    \par\smallskip\scriptsize (c) $\overline M$ begins the rover's route.
  \end{minipage}
  \caption{The Grasshopper trap and the start of the rover's route.
  Purple cells represent Mosquito supports;
  faded hexagons indicate continuations into the vertex structure.
  }
  \label{fig:trap-hive}
\end{figure}
\FloatBarrier

\paragraph{Leaving the vertex gadget.}
Releasing rovers is not sufficient: at least one must reach an edge
gadget, rather than remain confined to the vertex gadget.

After a lid deviation, the two successive rovers may originate from a
terminal-contact stack of the vertex gadget. If the guarded terminal
is occupied, one rover can use that contact to cross onto the edge row.
If it is empty, the first rover occupies it as a stationary bridge,
and the second crosses onto the edge row.

Unlike the Grasshopper-deviation case below, this passage is not
protected by lid ownership: the other lids at the terminal belong
to the deviating player. Any lid covering a rover receives the
prescribed permanent countercover from its source, exposing two
fresh rovers. A covered bridge remains occupied and continues to
provide the required contact. A fresh rover can reach the interception
cell directly
and resume the escape. Each interception therefore adds
only $O(1)$ moves and permanently consumes an attacking lid.
Consequently, repeated interceptions cannot prevent a rover from
reaching the edge network.

After a Grasshopper deviation, player $\overline{C}$ keeps the replacement
Grasshopper $\overline{G}$ at $t$.
The rover follows its vertex arm to this occupied
terminal and crosses onto the affected edge row.
The route uses stationary support, so reversals at the other
terminals cannot remove this path.
On this passage from the vertex gadget to an edge gadget, the only
stacks the rover must pass are the terminal-contact stacks, whose lids
are controlled by the rover's player and therefore cannot block it.
This is the only stack crossing before the rover reaches its target-edge
gadget.

Once a rover reaches an edge gadget, the global route result applies
(Section~\ref{sec:correctness-assembly}).
The passage through the vertex is included in the $O(N^2s)$ route bound.

The following lemma makes the reduction robust to deviations: the opponent
can prevent a player from winning by bypassing the simulated target-edge
objective.

\begin{lemma}
\label{lem:deviation-no-shortcut}
Let $C$ be the first player to deviate from a faithful
configuration, before either target edge has been reversed.
Suppose that $\overline{C}$ applies the prescribed immediate local reply
to every move requiring neutralization and introduces no
independent deviations. Between replies, $\overline{C}$ may make ordinary
edge or pass moves, or pursue the prescribed punishment.
Then any win by $C$ still requires its target-edge reversal
and the full traversal of the corresponding royal region.
\end{lemma}

\begin{proof}
The local analyses establish a confinement invariant at the
beginning of each turn of $C$: ordinary encoding Grasshoppers
are at their designated terminals or immobilized in traps;
covering-stack lids are at their original cells or neutralized;
and the remaining pieces of $C$ are immobile, except for the
separately analysed pass and royal pieces.

Every lid adjacent to a $C$-Grasshopper terminal belongs to $\overline{C}$.
Thus $C$ cannot move these lids to create new jump rays for its Grasshoppers, and
$\overline{C}$ leaves them in place. A lid deviation by $C$ occurs instead
at a $\overline{C}$-Grasshopper terminal and receives the prescribed local
reply. A Grasshopper deviation is answered by replacing the
vacated terminal and immobilizing the deviating Grasshopper.

The local analyses ensure that these replies remain legal,
including after earlier deviations and replies. Each response
preserves the occupied source supports and confines the displaced
pieces to the analysed neighbourhood. By separation, $\overline{C}$'s responses
affect neither $C$'s terminals nor the royal entrance.
The prescribed punishment moves also preserve the guards needed
for subsequent replies. Induction therefore maintains both the
availability of the responses and the confinement invariant,
although the position need no longer be faithful.

Suppose that a sequence of deviations enables a move of $C$
outside the analysed classes, and consider the first such move.
Immediately before it, the invariant holds and $\overline{C}$ has introduced
no independent deviations. Immobile pieces cannot move, and all lid
moves have already been analysed. Pass and royal
moves are covered in Sections~\ref{sec:pass} and~\ref{sec:royal}.
The purported move must therefore be made by an ordinary encoding
Grasshopper at a terminal. But its possible jumps remain those
already classified, since the configuration satisfies the invariant. This contradicts the choice of the
first unanalysed move.

Deviations may change which designated edge jumps satisfy the
vertex constraints, allowing reversals forbidden during faithful
play. They create neither a new class of moves nor another route
to the opposing Queen. In conclusion, $C$ must
still reverse its target edge and complete the full traversal
of Section~\ref{sec:royal}.
\end{proof}

% V4-CHECK COMPLETION: compare with the remaining traversal, not the
% full value W(R,L), if a target has already been reversed.

% The superseded trap and vertex sections of V3 were removed here.
\subsection{Face connectors and the One Hive property}
\label{sec:onehive}

The local AND and OR gadgets determine when a terminal Grasshopper may leave
without disconnecting a floating component. We must also ensure that all edge
gadgets form a single connected hive. Face connectors provide a fixed connected
backbone. 

\begin{definition}
\label{def:face-graph}
Let $\mathcal{G}$ be a plane graph and let $F(\mathcal{G})$ be its set of faces. The \emph{face
graph} $\mathcal{G}_F$ has vertex set $F(\mathcal{G})$; two distinct faces are adjacent when they
share an edge of $\mathcal{G}$ on their boundaries.
\end{definition}

\begin{lemma}
\label{lem:face-graph-connected}
The face graph $\mathcal{G}_F$ is connected, even when $\mathcal{G}$ is disconnected.
\end{lemma}

\begin{proof}
Join interior points of two faces by a curve that crosses edges transversely
and avoids vertices. The successive faces visited by the curve give a walk in
$\mathcal{G}_F$ between them.
\end{proof}

For each face $f$, we construct a \emph{connector gadget} $C_f$: a
connected network of flat Mosquitoes inside $f$. It meets the dead zone $D_e$
of each distinct boundary edge $e$ at one designated attachment away
from the terminals. For a polygonal edge, the attachment lies on just
one straight run. A bridge encountered twice along the same face
boundary needs only one attachment.

Since the connected components of $\mathcal{G}$ are placed side by side,
every bounded face has a connected boundary: otherwise, one component
would lie inside a bounded face of another, contrary to the layout.
Starting from any boundary edge, we route $C_f$ along the boundary walk at constant
cell-scale distance, with short branches to the designated dead
zones, as in Figure~\ref{fig:connector}. Following this walk reaches
every boundary edge of the face. For the outer face, we first build a track along the
outer boundary of each connected component of $\mathcal{G}$. We then join these
tracks inside the outer face to form a single connector.

The routing avoids terminal neighbourhoods, Grasshopper rays and
landings, floating components, traps, and reserved royal channels.
The attachments have the local separation required at each dead zone;
the figure shows only their global placement.

Every Mosquito of $C_f$ is adjacent only to flat Mosquitoes. It is therefore immobile by the same
Mosquito-only-adjacency argument as in the dead-zone observation
in Section~\ref{sec:edge-gadget}. This argument does not depend on
whether $C_f$ contains a cycle.

\begin{figure}[H]
\centering
\resizebox{0.8\linewidth}{!}{%
\begin{tikzpicture}
  \begin{scope}[x={(0.42cm,0cm)},y={(0.21cm,0.364cm)}]
    \tikzset{connector hex/.style={
      regular polygon,
      regular polygon sides=6,
      shape border rotate=30,
      minimum size=0.47cm,
      inner sep=0pt,
      outer sep=0pt,
      draw=white,
      line width=0.5pt
    }}

    \def\connectorhex#1#2#3{%
      \node[connector hex,fill=#3] at (#1,#2) {};
    }
    \def\connectorwhitehex#1#2#3{%
      \node[connector hex,draw=hivewhiteoutline,fill=#3]
        at (#1,#2) {};
    }

    % Vertex-gadget arms, lighter than the vertex cores.
    % The upper-left corner is a bend, not a fourth vertex.
    \foreach \q in {4,...,30} {
      \connectorhex{\q}{4}{hivepatch!55!white}
    }
    \foreach \r in {5,...,20} {
      \connectorhex{4}{\r}{hivepatch!55!white}
      \connectorhex{30}{\r}{hivepatch!55!white}
    }
    \foreach \q in {5,...,29} {
      \connectorhex{\q}{20}{hivepatch!55!white}
    }

    % Six player-colored cells per edge gadget.
    %
    % The gray attachments use the two internal positions:
    % one touches cells 2--3 and the other cells 4--5.
    % Hence the two terminal cells (1 and 6) are never touched
    % by a gray connector cell.

    % Bottom attachment: six black cells.
    \foreach \q in {14,...,19} {
      \connectorhex{\q}{4}{hiveedge}
    }

    % Left attachment: six white cells.
    \foreach \r in {10,...,15} {
      \connectorwhitehex{4}{\r}{hivewhiteedge}
    }

    % Right attachment: six black cells.
    \foreach \r in {9,...,14} {
      \connectorhex{30}{\r}{hiveedge}
    }

    % One closed connector, with a direct attachment
    % to each distinct edge.
    \foreach \q in {8,...,27} {
      \connectorhex{\q}{7}{hiveconnector}
    }
    \foreach \q in {8,...,26} {
      \connectorhex{\q}{17}{hiveconnector}
    }
    \foreach \r in {8,...,17} {
      \connectorhex{7}{\r}{hiveconnector}
    }
    \foreach \r in {7,...,16} {
      \connectorhex{27}{\r}{hiveconnector}
    }

    % Attachments of the closed connector.
    %
    % Bottom:
    % the unique gray attachment is placed in one of the two
    % internal positions, touching cells 2--3 rather than
    % being centered on the six-cell segment.
    %
    % Left:
    % the inner connector touches white cells 2--3.
    %
    % Right:
    % the inner connector touches black cells 4--5.
    \foreach \q/\r in {
      15/6,15/5,
      6/11,5/11,
      28/13,29/13
    } {
      \connectorhex{\q}{\r}{hiveconnector}
    }

    % Exterior connector tracks along the two oblique edges.
    % Two fading, complete hexagons at each end suggest continuation.
    % Their vertical extrema coincide with those of the vertex cores.
    \begin{scope}
      \foreach \q in {1,33} {
        \foreach \r in {5,...,19} {
          \connectorhex{\q}{\r}{hiveconnector}
        }
        \foreach \r in {4,20} {
          \node[connector hex,fill=hiveconnector,opacity=0.55]
            at (\q,\r) {};
        }
        \foreach \r in {3,21} {
          \node[connector hex,fill=hiveconnector,opacity=0.22]
            at (\q,\r) {};
        }
      }

      % Opposite-side attachments.
      %
      % Left exterior connector touches white cells 4--5.
      % Right exterior connector touches black cells 2--3.
      %
      % Thus the two gray connectors on either side of an edge
      % are close, but never meet the same player-colored cell,
      % and neither touches a terminal colored cell.
      \foreach \q/\r in {
        2/14,3/14,
        31/10,32/10
      } {
        \connectorhex{\q}{\r}{hiveconnector}
      }
    \end{scope}

    % Vertex cores.
    \foreach \x/\y in {4/4,30/4,30/20} {
      \foreach \dq/\dr in {0/0,1/0,-1/0,0/1,0/-1,1/-1,-1/1} {
        \pgfmathtruncatemacro{\qq}{\x+\dq}
        \pgfmathtruncatemacro{\rr}{\y+\dr}
        \connectorhex{\qq}{\rr}{hivepatch}
      }
      \fill[black!78] (\x,\y) circle[radius=0.075cm];
    }

    \node[font=\footnotesize,fill=white,inner sep=1.4pt]
      at (16,12.2) {connector gadget $C_f$};
    \draw[->,gray!65!black,line width=0.55pt]
      (16,11.55)--(18,7.5);

    \draw[<->,gray!70!black,line width=0.65pt]
      (21,4.55)--(21,6.45);
    \node[font=\scriptsize,anchor=west,fill=white,inner sep=0.4pt]
      at (21.6,5.35) {$d=O(1)$};
  \end{scope}
\end{tikzpicture}%
}
\caption{A closed face connector (gray) following the boundary of
a three-vertex face at constant cell-scale distance. At each attachment,
six cells simulating an edge gadget retain the player's colour. Pale purple
cells represent the vertex-gadget arms; darker purple patches mark the
vertex cores. Exterior connector tracks are shown along the two oblique
edges.}
\label{fig:connector}
\end{figure}
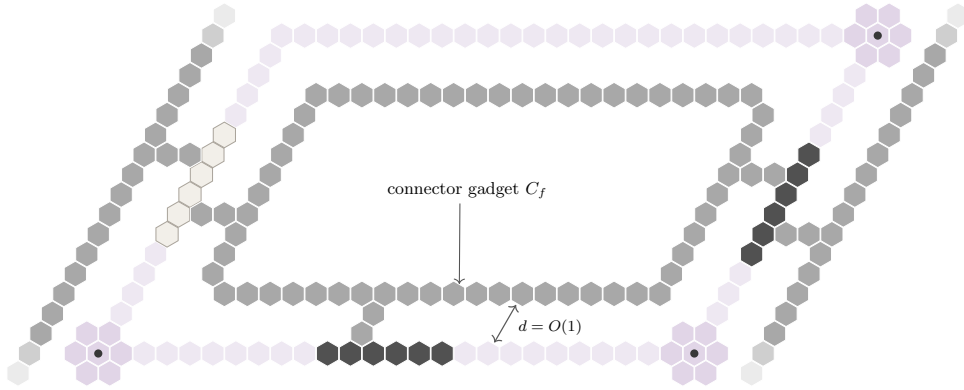

\begin{lemma}
\label{lem:onehive}
The union of all face connectors and edge-gadget rows is connected,
even when $\mathcal{G}$ is disconnected. Consequently, for any orientation
$\mathcal O$ of $\mathcal{G}$, the occupied support of the edge, vertex, and connector
gadgets encoding $\mathcal O$ is connected iff $\mathcal O$ is legal.
\end{lemma}

\begin{proof}
For each face $f$, the connected network $C_f$ meets the dead zone of every
distinct edge on its boundary, placing those edge rows in one component.
If two faces share an edge $e$, both connectors meet the connected dead zone
$D_e$ and hence lie in the same component. The face graph is connected by
Lemma~\ref{lem:face-graph-connected}, and every edge lies on the boundary of
at least one face. The outer-face connector joins the boundary components of
disconnected drawings. Therefore all connectors and edge rows lie in
one fixed exterior component.

The only remaining components are the floating components of
the AND and OR gadgets, whose only contacts with this exterior are their occupied
incident terminals. By Lemma~\ref{lem:vertex-gating}, all these components
are attached to at least one edge gadget if and only if every source inflow constraint
is satisfied. Thus the entire support is connected if and only if $\mathcal O$ is
legal. The same lemma ensures connectivity during each faithful edge
reversal.
\end{proof}

Thus, starting from a legal initial 2CL configuration, all structures
introduced so far belong to the connected component containing all the face
connectors. The remaining pieces will attach to edge gadgets,
so the complete initial \Hive{} position will also satisfy the One Hive rule.

The total connector length remains $O(N^2s)$, matching the bound for the
stretched edge arcs: the tracks follow these arcs at constant cell-scale
distance and with constant multiplicity, using only constant-length
attachment branches. The outer-face joins add only $O(Ns)$ cells, so the
total connector size is also $O(N^2s)$.
A rover reaching any edge gadget can therefore follow a simple upper path
through the connected edge--connector network to its own target-edge gadget.

%% file: Sec/RoyalRegions.tex
\section{Pass gadgets, royal regions and correctness}
\label{sec:royal-punishment}

The preceding gadgets already simulate legal edge reversals and would suffice
for a normal-play source game. The source formulation used here additionally
permits passes and specifies target edges. We now implement these features with
pass gadgets and royal regions, then assemble the position and compare the
global punishment time with the royal delay.

\begin{remark}[A normal-play alternative]
An EXPTIME-hard variant of unbounded 2CL on the same graph family,
with no passing or target edges and with the first player unable to
move losing, would require neither pass gadgets nor the royal countdown.
It would suffice to surround each Queen $Q^C$ with five covering stacks
controlled by $C$, leaving one neighbour empty and unreachable by
Grasshopper jumps. With no legal source move, $C$ must deviate and release
an opposing rover, which occupies that neighbour and surrounds $Q^C$. Hearn discusses replacing target edges with
this normal-play condition for 2CL~\cite[p.~53]{hearn2006thesis}.
Our attempts to prove EXPTIME-hardness for this variant were
unsuccessful, so we retain the target-edge formulation.
\end{remark}

\begin{figure}[H]
    \centering
    \includegraphics[width=1.0\linewidth]{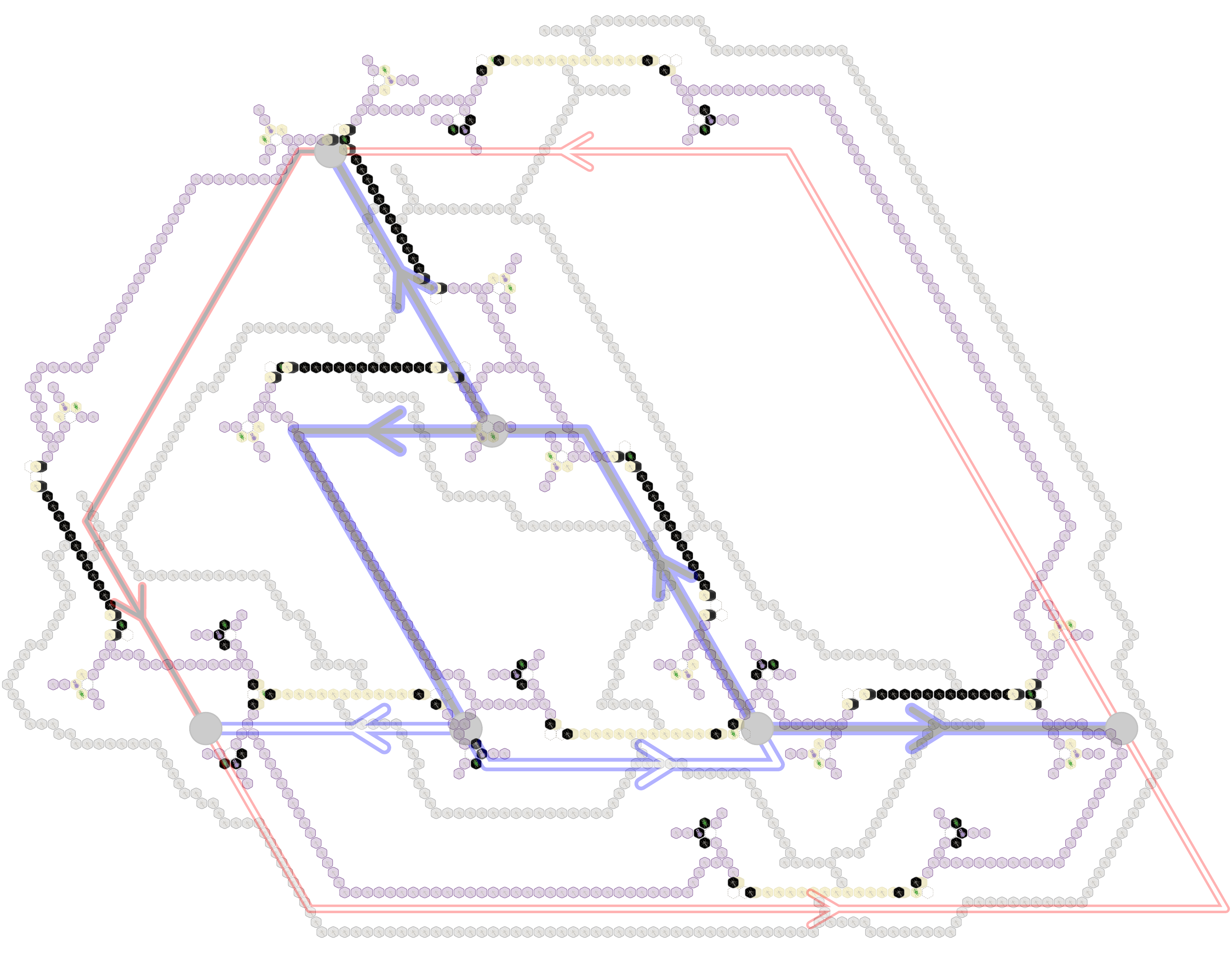}
    \caption{An example of the assembled vertex, edge, and connector gadgets,
    with a semi-transparent overlay of the oriented AND/OR graph embedding.}
    \label{fig:hive-instance}
\end{figure}

\subsection{Pass gadgets and unused pieces}
\label{sec:pass}

We realize optional source passes (Remark~\ref{rem:no-pass}) with one
gadget per player in the outer face (Figure~\ref{fig:pass-component}).
Five consecutive cells in a straight row contain flat Mosquitoes.
The two cells immediately beyond the row are the player's Grasshopper's
resting cells: it occupies one and jumps to the other, which is empty.
The outer-face connector attaches to the middle row cell, away from
both resting cells. Each resting cell touches only its row neighbour,
so no other jump ray is available.

These jumps preserve connectivity and change no encoded orientation.
The Mosquito beside the Grasshopper is a cut piece; all others see only
Mosquitoes. Thus the jumps are always available and are the gadget's
only moves.

\begin{figure}[H]
    \centering
    \begin{minipage}[c]{0.40\linewidth}
      \centering
      \includegraphics[width=\linewidth]{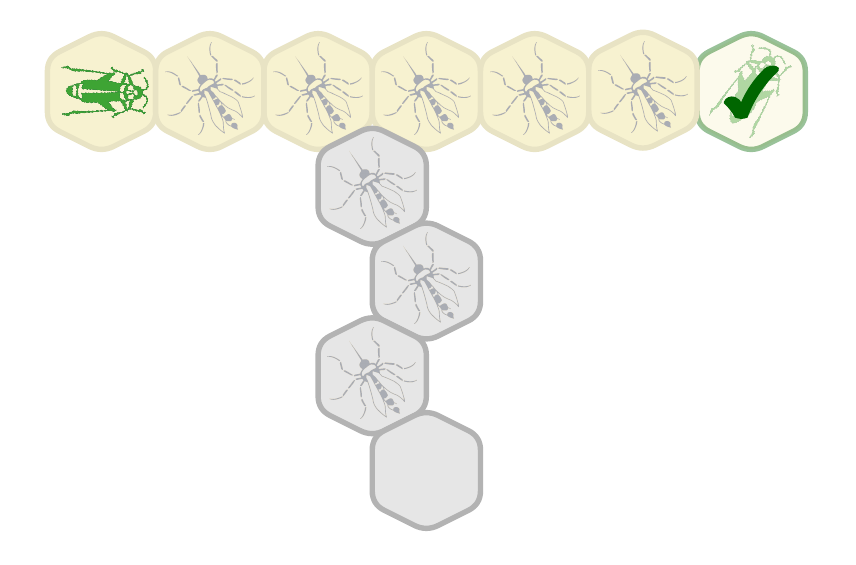}
    \end{minipage}\hfill
    \begin{minipage}[c]{0.50\linewidth}
      \centering
      \includegraphics[width=\linewidth,trim=0 0 0 70bp,clip]{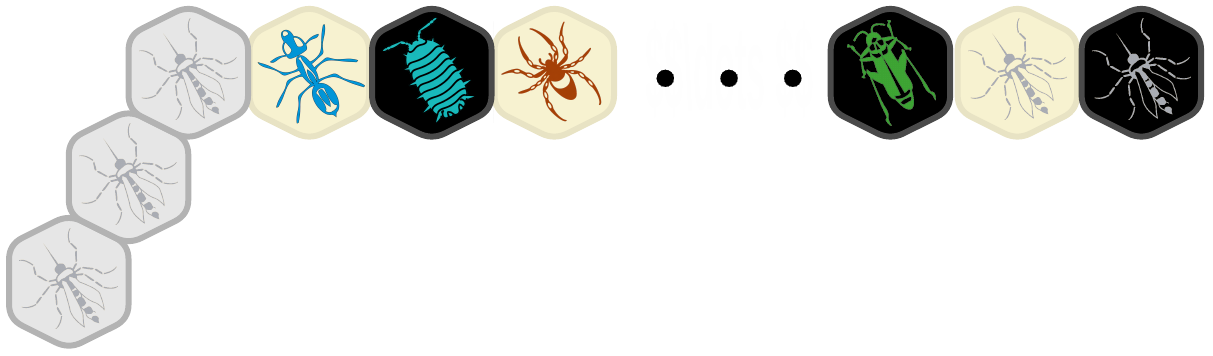}
    \end{minipage}
    \par\smallskip
    \begin{minipage}[t]{0.40\linewidth}
      \centering\scriptsize
      (a) The Grasshopper pass for White.
    \end{minipage}\hfill
    \begin{minipage}[t]{0.50\linewidth}
      \centering\scriptsize
      (b) Additional pieces with a two-Mosquito tail.
    \end{minipage}
    \caption{A pass gadget and an immobile inventory appendage.
    In the appendage, the additional pieces may have either colour. Gray cells belong to the outer-face connector.}
    \label{fig:pass-component}
    \label{fig:inventory-appendage}
\end{figure}

\FloatBarrier

\begin{remark}[Unused inventory pieces]
\label{rem:unused-inventory}
In games like generalized Chess, captures can leave a valid instance
without any pieces of a certain type. In \Hive{}, every inventory
piece must be already placed or in reserve, so omitting a required type
defines a different variant. Requiring further types or copies could
lower the complexity; EXPTIME-hardness for a restricted inventory
therefore does not automatically extend to such versions.

In our construction, however, arbitrarily many additional non-Queen
pieces can be neutralized in a straight
appendage joined to an outer-face connector by flat Mosquitoes
(Figure~\ref{fig:inventory-appendage}). The row starts directly after
this chain and ends with two flat Mosquitoes, preceded by a non-Pillbug
piece. Every appendage piece except the free-end Mosquito is a cut
piece; the endpoint sees only a Mosquito. Pillbugs and Mosquitoes
copying them can relocate only cut pieces, since the endpoint is
out of reach. Thus EXPTIME-hardness extends to every
inventory variant with enough Grasshoppers, Beetles, and Mosquitoes
to realize the construction.
\end{remark}

\subsection{Royal regions}
\label{sec:royal}

Reversing player $C$'s target edge sends its Grasshopper into what we call
$C$'s royal region, containing the opposing Queen $Q^{\overline{C}}$ and its
mate cell $F^{\overline{C}}$. The reversal does not immediately surround
$Q^{\overline{C}}$: it starts the \emph{completion phase}, a long sequence folded
into a serpentine of quadratic length and linear diameter. This delay
allows the punishment mechanisms of the preceding section to finish first.
We describe the region entered by White's target Grasshopper and shown in
Figure~\ref{fig:royal-region}; Black's region is its colour-swapped mirror.

Choose even parameters $R,L=\Theta(s)$, with $R\ge4$ and $L\ge8$.

\paragraph{The serpentine.}
The occupied support consists of $R$ parallel rows, each containing
$L+1$ cells. Consecutive rows are traversed in opposite directions
and joined by constant-size exterior turns. The ordinary supports are flat White Grasshoppers.

The traversing Grasshopper $G$ visits empty cells alternately on
the two sides of a row. Each jump crosses one occupied support cell
and advances to the next crossing position. A row has $L$ crossing
positions, numbered $0,\ldots,L-1$; an exterior turn contributes
one further landing before the next row begins. The even value of
$L$ and the alternating offsets keep the turns outside the rows.
Separation between folds prevents a turn landing from giving access
to the axis of another row, which would allow that row to be crossed
in a single jump.

At each row entrance, $G$ can touch three consecutive support cells.
Without an additional attachment, this contact could make the
middle support removable. We attach a short pendant branch to
that support, ending in the Beetle stack
\[
  (\mathrm{Bt}^{W},\mathrm{Bt}^{W},\mathrm{Bt}^{B}),
\]
listed from bottom to top. Removing the support would disconnect
the branch, so it remains a cut piece. We call these \emph{lateral stacks}.
They lie beside the serpentine and do not replace its crossing supports.

\paragraph{The entrance and intended continuation.}
The winning end of the target-edge gadget leads directly into the region,
without an arm extending to the corresponding vertex gadget: this edge
supplies no inflow there before its reversal, which ends the source game.

White's target reversal places $G$ at the first crossing position
$v_0$. This is the only entrance to the serpentine from the
surrounding construction. Its final jump reaches an initially
empty cell $v_f$.

The Black Queen $Q^B$ has five occupied neighbours and one empty
neighbour, the mate cell $F^B$. We call the constant-size structure
containing the Queen and these surrounding cells the \emph{keep}.
The five occupied neighbours contain copies of the covering stack
$MS^B$ defined in Section~\ref{sec:edge-gadget}, with Black lids and
White bodies. While these neighbours remain occupied, Freedom to Move
prevents the Queen from sliding to $F^B$, since both flanking cells are
occupied.
A straight ramp of White Grasshoppers
$G_1,G_2,G_3,\ldots$ points towards $F^B$. The first Grasshopper, $G_1$,
is covered by a Black Beetle.
This is the only Beetle-over-Grasshopper stack in the serpentine
and ramp. The other ramp Grasshoppers are initially cut pieces.
Although their forward jump ray leads to $F^B$, One Hive prevents
them from leaving.

The last occupied cell crossed by $G$ carries the exit stack:
a Black Beetle lid above a constant-size White Beetle body.
The body keeps this cell occupied when an exposed White Beetle
leaves it. The cell $v_f$ lies beside this stack and the start of the ramp.

Before the terminal connection, with the stacks undisturbed,
One Hive pins the serpentine supports and the uncovered ramp
Grasshoppers. At row entrances, the pendant branches preserve
this pinning despite the additional contacts made by $G$.
The Beetle bodies and $G_1$ remain covered.
The Black Beetle lids can move; we discuss their moves below.

Occupying $v_f$ joins the serpentine to the ramp and closes a cycle
through the keep and its external connection. This releases $G_2$
and the later ramp Grasshoppers. A released Grasshopper jumps along
the occupied ramp to $F^B$, surrounding $Q^B$.
The role of $G$ is to maintain the connection at $v_f$; the winning
jump is made by a Grasshopper already on the ramp.

The neighbouring empty cell $v_f^\star$ provides an alternative
connection at $G_3$. It is not another landing of the prescribed
countdown. Its name is used below to distinguish the possible
exit-lid moves.

The keep attaches to the target-edge gadget through a connector
of constant length. The long side of the region is the Grasshopper
ramp, traversed by the winning jump. External pieces are separated
from the serpentine except at the declared interfaces, so they
create neither additional contacts nor additional jump landings.

\begin{figure}[tb]
    \centering
    \includegraphics[width=0.95\linewidth]{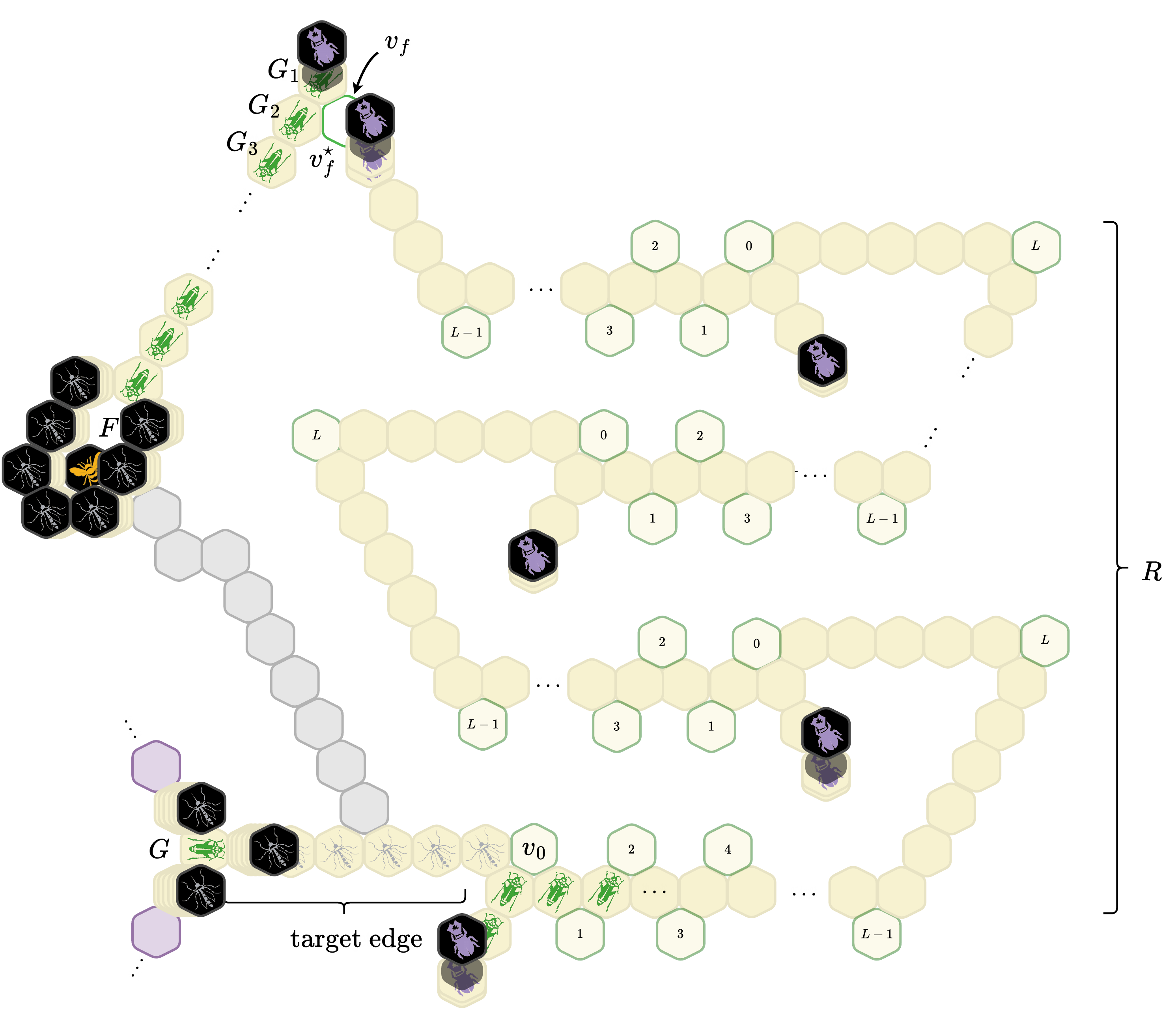}
    \caption{White's royal region, containing $Q^B$ and $F^B$ (labelled $F$).
    The numbers indicate the target Grasshopper's traversal order along each row.}
    \label{fig:royal-region}
\end{figure}

\paragraph{Lateral and terminal lid moves.}
If Black moves a lateral lid, the upper White Beetle underneath
covers it on White's next move and remains there as a guard.
The lower White Beetle stays at the source, preserving the
attachment. This also permits the reply when the lid descends
into an empty cell beyond the pendant branch.

Each lateral lid is permanently covered, and the lid move and its
countercover consume one turn per player without changing the rest
of the construction. Suppose the defender has a winning or drawing
strategy $\tau$. We obtain a strategy with the same guarantee by
omitting these pairs. Indeed, any attacker strategy against the
resulting strategy can be extended to one against $\tau$ by inserting
the prescribed countercovers, keeping the guards fixed, and leaving
the newly exposed Beetles unused. Since the pairs preserve the player
to move and all existing continuations, the corresponding plays have
the same outcome. Moreover, only finitely many pairs can occur, as
each lateral lid moves at most once. Thus the defender may avoid
lateral lid moves without loss of generality and therefore the strategies constructed below keep their own lateral lids fixed.

This argument applies both before the target reversal and during
the countdown. The prescribed response requires no punishment rover.

Moving the exit lid into $v_f$ immediately creates the connection
that releases $G_2$; moving it into $v_f^\star$ releases $G_3$.
In either case White has the corresponding winning ramp jump.

If the exit lid moves elsewhere, White uses the exposed body
Beetle to occupy either $v_f$ or $v_f^\star$, leaving the exit cell
occupied. Regardless of the position chosen, the displaced lid cannot cover both $G_2$ and $G_3$ on Black’s next turn. Thus White wins.

The guard of $G_1$ can move, but must keep $G_1$ immobile,
either covered or pinned by One Hive. Such harmless moves represent
source passes before either target reversal (and provide the pass
mechanism in the Mosquito-free reduction of
Section~\ref{sec:beetle-stacks}). Freeing $G_1$ gives White
an immediate winning jump along the ramp to $F^B$, even before
the target reversal.
After the terminal connection, the guard cannot prevent the win:
covering $G_2$ frees $G_1$, whereas keeping $G_1$ pinned leaves
a released ramp Grasshopper available for the winning jump.
Thus the terminal mechanism uses Beetles and Grasshoppers
throughout; it requires no Mosquito movement ability. After the terminal connection, some ordinary supports may cease
to be cut pieces. They all belong to White, which leaves them
in place. Their release therefore gives Black no additional
move that can break the connection.

\begin{lemma}
\label{lem:countdown-progress}
The first player to enter its royal region during faithful play
can force a win in at most $W(R,L)=R\cdot(L+1)$ further moves
of its own, excluding immediate replies to neutralized deviations:
each such deviation--reply pair consumes one turn of each player
without advancing towards a win.
\end{lemma}

\begin{proof}
By colour symmetry, assume White enters first.
The prescribed route contains $RL$ crossing positions,
$R-1$ intermediate turn landings, and $v_f$.
Since the target reversal has already placed $G$ at $v_0$,
reaching $v_f$ takes $R(L+1)-1$ further moves.
One final ramp jump occupies $F^B$.
Fold separation prevents shortcuts, and One Hive pins the ramp
until the terminal connection is made.

White keeps its ordinary supports fixed. The local responses above
and in Section~\ref{sec:deviations-traps} either yield an earlier win
or neutralize a deviation in one move; harmless deviations need no
reply. Each neutralized deviation--reply pair consumes one turn per
player, does not advance Black towards a win, and cannot recur.
Omitting these finitely many pairs preserves White's lead in the race between
equal-length royal completions, so White wins first.
\end{proof}

With $R,L=\Theta(s)$, the countdown has length $\Theta(s^2)$.
The support rows use $O(RL)$ pieces; the turns and constant-height
stacks add $O(R)$, and the ramp adds $O(R+L)$.
The whole region therefore uses $O(s^2)$ pieces and fits inside
a parallelogram whose side lengths are $O(s)$. By the embedding
construction, every face contains a parallelogram with side
lengths $\Theta(s)$ (Lemma~\ref{lem:layout}), providing space for the
royal region, including its longer turns.

The race in Lemma~\ref{lem:countdown-progress} compares remaining
completion times. By contrast, the bound $P_B(N,s)<W(R,L)$ below
compares punishment with a full traversal before the deviator
enters its royal region.

% V4-CHECK ROYAL COMPLETION: the support-colour issue is addressed:
% ordinary supports belong to the completing player, who keeps them fixed.
% Still verify the response to deviations during completion using remaining
% traversal times, and justify that neutralized lateral pairs preserve
% relevant jump rays, gates, and the residual routes of both royal regions.

\subsection{Assembly and distance bounds}
\label{sec:correctness-assembly}

We now join the local gadgets and compare the global punishment route
with the royal countdown.

\begin{lemma}
\label{lem:global-assembly}
The gadgets can be assembled in polynomial time into a single \Hive{} position
such that
\begin{enumerate}[label=\textup{(\roman*)}]
\item faithful configurations satisfy the One Hive rule exactly when their
      encoded source orientations are legal;
\item no connector gadget changes a vertex-gadget contact or a faithful Grasshopper
      landing;
\item every trap release point has an upper route to the appropriate keep, of
      length $O(N^2s)$.
\end{enumerate}
\end{lemma}

\begin{proof}
Lemma~\ref{lem:layout} supplies disjoint strips along the stretched edge arcs,
the required local connections to the vertex gadgets, and the necessary
clearance. By Lemma~\ref{lem:onehive}, the edge, vertex, and connector gadgets
have connected support if and only if the source orientation is legal. The pass gadgets
and royal regions are internally connected and attach to the skeleton without
creating additional contacts with the floating components. These additions
therefore preserve the equivalence, proving (i). The same clearance leaves
vertex contacts and faithful Grasshopper landings unchanged, proving (ii).
The skeleton reaches each keep through the constant-length
connector beside its target-edge gadget. A gate can be bypassed
by climbing a lateral tower and descending towards the next cell,
at constant additional cost. Since the stretched edge arcs have
total length $O(N^2s)$, the upper routes have the same order.
\end{proof}

\begin{lemma}
\label{lem:global-punishment-final}
If player $C$ is the first to deviate before either target edge
has been reversed, player $\overline{C}$ can surround $Q^C$
within
\[
  P_{\overline{C}}(N,s)=O(N^2s)
\]
of its own moves, regardless of $C$'s subsequent play.
\end{lemma}

\begin{proof}
A first deviation in a royal region gives $\overline{C}$ a win
in a constant number of moves by Section~\ref{sec:royal}.
Otherwise, the local responses of Section~\ref{sec:deviations-traps}
neutralize the deviation and provide a rover that can reach
the edge network. By Lemma~\ref{lem:global-assembly}, this rover
has a route of length $O(N^2s)$ to the mate cell $F^C$,
including its escape from the releasing gadget.
It retains Beetle movement until its final descent.

Along this route, the rover encounters covering stacks only
at vertex gadgets and at the keep. The local escape analysis in
Section~\ref{sec:deviations-traps} handles both protected passages
after Grasshopper deviations and possible interceptions after
lid deviations, at $O(1)$ additional cost per attacking lid.
At the keep, the prescribed local response releases a rover
already within constant distance of $F^C$, allowing
$\overline{C}$ to finish in $O(1)$ further moves.

Player $\overline{C}$ answers each further deviation before
resuming the route, keeping all prescribed guards and terminal
replacements in place. The local analyses ensure that these
responses remain available, permanently immobilize the deviating
pieces, and preserve the rover's escape route. Any affected
$\overline{C}$-Grasshoppers can simply be left fixed.
Even if vertex constraints are relaxed,
Lemma~\ref{lem:deviation-no-shortcut} still requires $C$
to reverse its target edge and complete the full royal traversal.
The royal responses of Section~\ref{sec:royal} likewise
neutralize an interruption or give an earlier win.

There are $O(N)$ covering-stack lids and encoding Grasshoppers,
and $O(R)$ lateral royal lids. Each can cause at most one
neutralized interruption, at constant cost. Hence, for suitable
absolute constants $C_1,C_2,C_3$,
\[
  P_{\overline{C}}(N,s)
  \le C_1N^2s+C_2(N+R)+C_3
  =O(N^2s),
\]
since $R=\Theta(s)$. This bound includes all subsequent
replies and the constant-cost completion at the keep.
\end{proof}

\begin{corollary}
\label{cor:timing-separation}
The stretch factor $s$ can be chosen polynomially so that
\begin{equation}
\label{eq:punishment-before-win}
  P_B(N,s)<W(R,L).
\end{equation}
\end{corollary}

\begin{proof}
By Lemma~\ref{lem:countdown-progress} and the choices $R,L=\Theta(s)$ in
Section~\ref{sec:royal}, an undisturbed royal completion takes exactly
$W(R,L)=R(L+1)$ moves and satisfies $W(R,L)\ge as^2$ for some constant $a>0$.
Lemma~\ref{lem:global-punishment-final} gives $P_B(N,s)\le AN^2s$ for another
constant $A>0$. Let $m$ absorb the fixed divisibility requirements of the gadgets
and choose
\[
  s=m\left\lceil\frac{KN^3}{m}\right\rceil
\]
with $K>A/a$ and large enough for all geometric requirements. Since
$N\ge1$, we have $AN^2s<as^2\le W(R,L)$.
This compares punishment with a full royal traversal; it applies
when the deviating player has not already entered that traversal.
The choice $s=O(N^3)$ is polynomial.
\end{proof}

\subsection{Correctness and complexity}
\label{sec:reduction-correctness}

We now prove that the resulting position preserves winning strategies. Let $I$ be a
source instance, and let $\mathcal{G}$ be its constraint graph. Recall that
$N=|V(\mathcal{G})|+|E(\mathcal{G})|$.

Construct the position described in Sections~\ref{sec:embedding}--\ref{sec:royal}
using the value of $s$ fixed above and the same parameters $R,L$ in both royal
regions. The layout and gadget size bounds give a total piece count of
\[
  O(N^2s+s^2),
\]
which is $O(N^6)$ for $s=\Theta(N^3)$. All coordinates, colours, types, and stack
orders are produced explicitly, so the construction takes polynomial time.

We balance the inventories using the immobile appendage of
Remark~\ref{rem:unused-inventory}, leaving no pieces in reserve.
The added pieces are linear in the construction size and do not
affect the simulation.

\begin{theorem}[Correctness of the reduction]
\label{thm:reduction-correctness-final}
White can force its target reversal in $I$ if and only if White can force a win
in the constructed \singleQueenHive{} position.
\end{theorem}

\begin{proof}[Completeness]
Let $\sigma$ be a winning strategy for White in $I$. White implements each move prescribed by $\sigma$ with the corresponding
designated \Hive{} move; edge and pass jumps implement legal source moves by
Lemmas~\ref{lem:vertex-gating} and~\ref{lem:global-assembly} and
Section~\ref{sec:pass}; harmless royal-guard moves also represent passes.
Whenever Black makes a neutral interruption, White answers with the
prescribed countercover. This pair changes no encoded orientation and
restores the source position and player to move, so White continues to
follow $\sigma$ from that state.

Before either target reversal, the remaining royal deviations give White an early win by
Section~\ref{sec:royal}. If Black deviates in an edge or vertex gadget
before either target reversal, White punishes by
Lemma~\ref{lem:global-punishment-final}. Black still needs its full
royal traversal by Lemma~\ref{lem:deviation-no-shortcut}, so
\eqref{eq:punishment-before-win} ensures White wins first.

Otherwise $\sigma$ brings White into its royal region first, and
Lemma~\ref{lem:countdown-progress} gives a forced win against all
subsequent Black play.
\end{proof}

\begin{proof}[Soundness]
We prove the contrapositive. Suppose White cannot force its target reversal
in $I$, and fix a Black strategy $\tau$ that prevents it.
Black implements each move prescribed by $\tau$ with its corresponding
designated Hive move. Whenever White makes a neutral interruption,
Black answers with the prescribed countercover. This pair changes no
encoded orientation and restores the source position and player to move,
so Black continues to follow $\tau$ from that state, treating harmless
royal-guard moves as passes.

Before either target reversal, the remaining royal deviations give Black an early win. If White
deviates in an edge or vertex gadget before either target reversal,
Black punishes by Lemma~\ref{lem:global-punishment-final}.
Lemma~\ref{lem:deviation-no-shortcut} requires White's full traversal,
while \eqref{eq:punishment-before-win} lets Black win sooner.

Otherwise $\tau$ prevents White's target reversal. If Black reverses
its own target, Lemma~\ref{lem:countdown-progress} gives a Black win;
if neither target is reversed, White never reaches a winning completion.
Thus White cannot force a win in the constructed position.
\end{proof}

Since only finitely many interruptions can occur, an infinite
faithful play projects to an infinite source history. If White's target
is never reversed, it is not a White win, under the reachability
convention of Definition~\ref{def:decision}.
Passing causes no mismatch: Section~\ref{sec:pass} represents each optional source
pass by a genuine Grasshopper move, so the automatic no-move pass of \Hive{} is
never used during faithful play. During punishment, an automatic pass by the
deviating player can only give the punisher another turn.

By Theorems~\ref{thm:2cl-exptime} and~\ref{thm:reduction-correctness-final},
the polynomial-time reduction above establishes EXPTIME-hardness
for \singleQueenHive{}.

\begin{lemma}
\label{lem:hive-membership-final}
Deciding whether White can force a win is in EXPTIME for
\singleQueenHive{} and both multi-Queen variants of
Definition~\ref{def:multi-queen-variants}.
\end{lemma}

\begin{proof}
Let $n$ count all inventory pieces per player, including reserves. Normalize
positions by translating the lexicographically least
occupied cell to the origin. One Hive places every occupied cell
within lattice distance $2n-1$ of the origin. The board, ordered
stacks, reserves, player to move, and preceding-turn record therefore
use $O(n\log n)$ bits, even when most pieces are in reserve.
Since inventories are listed explicitly, $n$ is at most linear in
the input size, and there are at most $2^{O(n\log n)}$ normalized states.

There are polynomially many candidate placements, destinations, and
special-ability actions; their legality and terminal outcomes can be
checked in polynomial time in $n$. Allowing more Queens changes neither
the state bound nor move generation; testing whether any or all Queens
are surrounded is also polynomial.

Standard backward propagation of winning states solves the resulting
finite turn-based reachability game in time polynomial in its graph
size, with draws and infinite play treated as non-wins. Hence all
three decision problems are in EXPTIME.
\end{proof}

\begin{theorem}
\label{thm:single-queen-complete-final}
The \singleQueenHive{} decision problem is EXPTIME-complete,
even when all pieces are already placed.
This holds for every inventory variant with enough Grasshoppers,
Beetles, and Mosquitoes to realize the construction, with any
additional non-Queen pieces neutralized.
\end{theorem}

\begin{proof}
The reduction above and Remark~\ref{rem:unused-inventory} give hardness
under the stated restrictions, and
Lemma~\ref{lem:hive-membership-final} gives membership.
\end{proof}

\begin{corollary}
\label{cor:multi-queen-complete-final}
Both multi-Queen variants of Definition~\ref{def:multi-queen-variants}
are EXPTIME-complete under the same placement and inventory restrictions.
\end{corollary}

\begin{proof}
Every \singleQueenHive{} instance is also an instance of both
multi-Queen variants. With one Queen per player, the legal moves and
both loss conditions coincide, so the identity map transfers the
hardness of Theorem~\ref{thm:single-queen-complete-final}, including
the stated placement and inventory restrictions.
Lemma~\ref{lem:hive-membership-final} supplies the common upper bound.
\end{proof}

%% file: Sec/BeetleStacks.tex
\section{Replacing Mosquitoes with Beetle stacks}
\label{sec:beetle-stacks}

Eliminating the Mosquito, an expansion piece, establishes
EXPTIME-completeness for single-Queen $n$-Hive using only base-game
pieces. We presented the Mosquito construction first because it
simplifies both the proof and the diagrams; we now explain how
to obtain the stronger result.

Most Mosquitoes in our reduction serve as immobile structural supports:
a flat Mosquito adjacent only to Mosquitoes has no movement ability to
copy. With Beetles, this immobility must instead be enforced
strategically. Each support becomes a stack whose lid belongs to
one player and whose body belongs to the other. Moving the lid
exposes an opposing Beetle, which covers it, while further body
Beetles carry out the punishment.

Two points require proof. Covering stacks now occur throughout the
construction, and changes in neighbouring stack heights can make a
countercover illegal under the Freedom to Move rule. Rovers can also be intercepted
anywhere along their route, rather than only during a vertex gadget
escape. We first ensure that every lid departure admits a countercover,
then bound the delay caused by repeated interceptions.

\subsection{Replacement and immediate responses}
\label{sec:beetle-replacement}

We replace structural Mosquitoes in the edge rows, dead zones,
vertex and pass gadgets, and connectors by Beetle stacks in the
same cells. Existing covering stacks retain their lid colours.
Write $h(x)$, abbreviated $h_x$, for the height of cell $x$ in the
position under consideration, and $h_x^0$ for its initial height.
A replacement stack with lid colour $C$ has the form
\[
  BS^C(x)=
  \bigl(
    \mathrm{Bt}^{\overline{C}}_1,\ldots,
    \mathrm{Bt}^{\overline{C}}_{h_x^0-1},
    \mathrm{Bt}^{C}_{\mathrm{lid}}
  \bigr),
\]
written from bottom to top, with $h_x^0$ chosen below.

The rover Mosquito in each Grasshopper trap becomes a single Beetle
of the same colour. It already has the movement ability formerly
obtained by copying, and there are no remaining structural Mosquitoes
to which it could supply that ability.

These changes preserve the occupied support, so faithful Grasshopper
landings and the connectivity tests of the vertex gadgets remain
unchanged.

\paragraph{Lid colours and local responses.}
For an encoding Grasshopper of colour $C$, the covering stacks at its
terminals retain lids of colour $\overline{C}$; we give the same colour
to the replacement lids in its Grasshopper traps. Thus $C$ cannot
move these lids to change jump rays or trap landings, or to reconnect
the pendant support that pins a trapped Grasshopper. Its opponent
keeps them fixed when applying the trap response.

Conversely, a lid departure by $\overline{C}$ near a terminal can change the
available jumps of a Grasshopper belonging to $C$.
Player $C$ need not use those jumps. It leaves that Grasshopper
fixed and pursues punishment with the released Beetles.
The countercover can therefore remain permanently above the displaced
lid; it never has to transfer onto the Grasshopper.

The vertex gadget escape retains its two-rover provision. When the
terminal is empty, one rover occupies it as a stationary bridge and
a second crosses onto the edge row. When the terminal is occupied,
a rover can cross using that existing contact.
The royal mechanisms retain their separately analysed local responses.
The replacement network considered below consists of the edge,
vertex, pass, and connector gadgets.

\paragraph{Freedom to Move above the hive.}
We recall the Freedom to Move rule from the preliminaries.
For a Beetle step from $x$ to an adjacent cell $y$, take $h_x,h_y$
before the move, including the moving Beetle at $x$, and let $f_1,f_2$
be the heights of the two flanks. The travel level and height test are
\[
  L=\max\{h_x,h_y+1\},\qquad \min\{f_1,f_2\}<L.
\]
A single lower flank therefore suffices, as illustrated in
Figure~\ref{fig:beetle-height-gates}. 

\begin{figure}[H]
  \centering
  \begin{minipage}[t]{.245\linewidth}
    \centering
    \includegraphics[width=.70\linewidth]{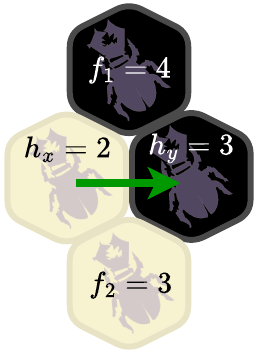}\par
    {\scriptsize (a) Climb allowed: $L=4$.\par}
  \end{minipage}\hfill
  \begin{minipage}[t]{.245\linewidth}
    \centering
    \includegraphics[width=.70\linewidth]{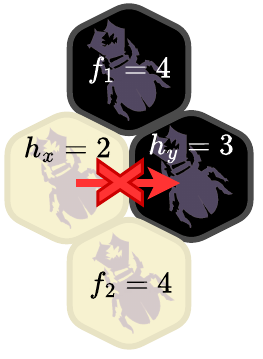}\par
    {\scriptsize (b) Climb blocked: $L=4$.\par}
  \end{minipage}\hfill
  \begin{minipage}[t]{.245\linewidth}
    \centering
    \includegraphics[width=.70\linewidth]{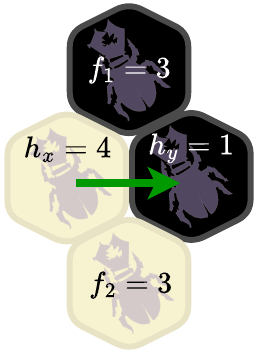}\par
    {\scriptsize (c) Descent allowed: $L=4$.\par}
  \end{minipage}\hfill
  \begin{minipage}[t]{.245\linewidth}
    \centering
    \includegraphics[width=.70\linewidth]{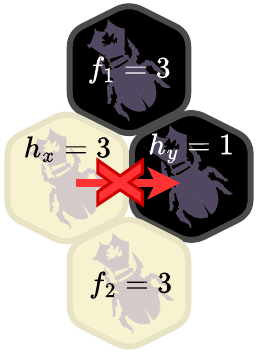}\par
    {\scriptsize (d) Descent blocked: $L=3$.\par}
  \end{minipage}
  \caption{Freedom to Move above the hive, shown from above.
  Numbers give stack heights before the move, including the moving
  Beetle at $x$.}
  \label{fig:beetle-height-gates}
\end{figure}

\paragraph{Geometry and initial heights.}
The staggered connector attachments ensure that every structural
cell has at most three initially occupied neighbours.
In particular, each replacement stack is initially adjacent to
at most three lid-bearing stacks.

For our purposes, it will be enough to keep a fixed gap between the initial
heights of adjacent stacks. The triangular lattice admits a proper
three-colouring, which lets us do this with three height classes.
In axial coordinates
$(q,r)$ with basis
vectors at $60^\circ$, set
\[
  \chi(q,r)=(q-r)\bmod 3.
\]
The neighbour steps $\pm(1,0)$, $\pm(0,1)$, and $\pm(1,-1)$
change $q-r$ by $\pm1$ or $\pm2$. Applying this colouring to the
assembled construction, we assign
\[
  h_x^0=16\bigl(1+\chi(x)\bigr)\in\{16,32,48\}.
\]
Adjacent stacks thus have distinct initial heights, independently
of their lid colours.

\paragraph{The response strategy.}
Let $C$ be the first player to deviate and let $\overline{C}$ be its opponent.
Player $\overline{C}$ follows this strategy:
\begin{enumerate}
  \item
  After a lid departure, $\overline{C}$ moves the newly exposed Beetle from
  the same source onto that lid. This countercover remains in
  place permanently.
  \item
  After an unintended Grasshopper jump, $\overline{C}$ uses the original trap
  response: restore the vacated terminal, pin the deviating
  Grasshopper by One Hive, and release the rover. The replacement
  Grasshopper stays at the restored terminal.
  \item
  Between responses, $\overline{C}$ keeps its original lids and ordinary
  encoding Grasshoppers fixed. It moves at most one travelling
  rover, using initially occupied cells except for a prescribed
  terminal bridge or the final winning descent.
  \item
  After supplying its countercover, each source supplies at most
  two additional body Beetles. One may remain as a stationary
  bridge at an empty terminal, while the other travels.
  No terminal receives more than one such stationary bridge.
\end{enumerate}

\paragraph{Changes in height.}
Assume all earlier lid departures have received permanent countercovers,
so each original lid has moved at most once. A replacement stack
loses at most four pieces: its lid, the countercover, and two further
body Beetles. Each arriving lid adds itself and its countercover.
If it covers a rover, that rover also becomes permanent, giving
a charge of $+3$ rather than $+2$.
At most three lids arrive, contributing at most nine permanent pieces;
one travelling rover may also be present. Hence
\[
  h_x^0-4\le h_x\le h_x^0+10.
\]

We now bound the heights of cells that are initially empty and contain no replacement stack.

Each terminal has at most three original lid sources.
Allowing one ground piece, one stationary bridge, at most three
permanent pieces per arriving lid (including a captured rover),
and one travelling rover gives height at most
\[
  1+1+3\cdot3+1=12.
\]
The protected trap cells may contain a ground piece and a rover,
but their fixed lids prevent opposing lid arrivals.
All remaining cells without replacement stacks receive only
lid--countercover pairs, giving height at most $2\cdot6=12$.
Thus every cell without a replacement stack has height at most $12$.

\begin{lemma}[Permanent responses to deviations]
\label{lem:beetle-counterstrategy}
Let $C$ make the first deviation from faithful play in the replacement
network, and let $\overline{C}$ be its opponent.
Player $\overline{C}$ has a response strategy that, regardless of $C$'s subsequent
play, neutralizes every subsequent lid departure and unintended
Grasshopper jump in this network in one move.
Each neutralized piece remains immobilized for the rest of play
under this strategy. Immediately after a countercover, the next
exposed body Beetle also has a legal step onto the countercovered
column.
\end{lemma}

\begin{proof}
We prove inductively that $\overline{C}$ can follow the prescribed responses.
Assume all earlier deviations have been neutralized.
The height bounds above then apply, and a $C$-lid that can still
move remains at its original cell and initial height.

Let such a lid move from $x$ to $y$, with pre-move heights $h_x,h_y$;
in particular, $h_x=h_x^0$. We analyse the travel levels of the lid move, its countercover,
and the next exposed Beetle's step onto the countercover.
Their respective travel levels are
\[
  L_j=\max\{h_x-j,h_y+j+1\},\qquad j=0,1,2.
\]
The flank heights $f_1,f_2$ are unchanged by these steps.
For $j=1,2$, if $L_j\ge L_0$, the legal lid move already guarantees
clearance. Otherwise $L_0=h_x$: in this case the two inequalities are
\[
  \min\{f_1,f_2\}<h_x \qquad\text{and}\qquad
  \min\{f_1,f_2\}\ge L_j\ge h_x-j.
\]
The first follows from the legality of the lid move; the second is the
condition for Freedom to Move to block the countercovering Beetle's step.

The height separation excludes this obstruction. A flank occupied by a replacement stack has an initial height
at least $16$ below or above $h_x$.
Its current height is therefore at most $h_x-6$ in the former case
and at least $h_x+12$ in the latter.
A flank without a replacement stack has height at most
$12\le h_x-4$. Thus every flank below $h_x$ is at most $h_x-4$,
which is below both follow-up levels.
Both steps also satisfy One Hive: the source remains occupied
and the destination already contains the displaced lid.
The countercover stays permanently above that lid, including when
the rover subsequently crosses the column.

Grasshopper deviations use the trap response already described.
Both responses preserve the strategy's conditions, completing the induction.
\end{proof}

\subsection{Upper route and punishment}
\label{sec:beetle-punishment}

Let $C$ be the first player to deviate before either target reversal,
and let $\overline{C}$ apply the response strategy above.
After the local escape, $\overline{C}$ has a rover on the connected
upper network leading to the open neighbour $F^C$ of $Q^C$.

Height gates do not prevent traversal of this network.
If a step from $x$ to $y$ is blocked, the rover can instead
climb onto either blocking flank $z$ and then descend to $y$.
Both steps satisfy Freedom to Move: $y$ provides a lower flank
for the climb, and the vacated column $x$ provides one for
the descent. One Hive is preserved because the underlying
support cells remain occupied.

These detours still use cells of the same network, which contains
$O(N^2s)$ cells. Removing cycles from the resulting walk therefore
gives a simple legal route of length $O(N^2s)$.
The separated height ranges established above preserve the
relative order of adjacent replacement columns, while all other
columns remain too low to obstruct these passages.
Thus subsequent local replies do not invalidate the chosen route;
interceptions of the rover are handled below.

Including the local escape and final royal continuation,
the route still takes $O(N^2s)$ moves.
The rover traverses neither the serpentine nor the long ramp.

\begin{lemma}[Continuing the upper route]
\label{lem:beetle-duel}
Covering the rover and receiving its countercover leaves $\overline{C}$ with
a fresh rover whose legal upper continuation is at most one step
longer. This also holds if the captured rover had not yet left its
release cell.
\end{lemma}

\begin{proof}
Let $v$ be the captured rover's cell and $w$ the attacking lid's
source. Lemma~\ref{lem:beetle-counterstrategy} supplies a permanent
countercover and a fresh rover that can step from $w$ to $v$.
The old continuation from $v$ remains legal, so a route of length
$\ell$ is replaced by one of length at most $1+\ell$.
This does not require any prior move by the captured rover.
\end{proof}

\paragraph{Bounding the delay.}
Player $\overline{C}$ answers each deviation immediately and otherwise advances
its rover along the route. Faithful edge reversals do
not change the connector network.

Let $D=O(N^2s)$ bound the initial route together with the prescribed
local manoeuvres. Let $I$ count interceptions of the rover,
and let $J$ count the neutralizing replies,
including the initial reply.
Each ordinary advance decreases the
remaining route length by one, while Lemma~\ref{lem:beetle-duel}
allows an increase of at most one per interfering lid. Thus the
total number of $\overline{C}$-moves needed for punishment is at most
\[
  D+I+J.
\]
This counts both costs of an interception: the immediate countercover
and the possible extra step by the replacement rover.

Let $T=O(N^2s)$ be the number of $C$-lids in the replacement network,
and let $G=O(N)$ be the number of encoding Grasshoppers.
By Lemma~\ref{lem:beetle-counterstrategy}, each deviating piece is
permanently immobilized after its first deviation. Consequently,
\[
  I\le T,\qquad J\le T+G,
  \qquad D+I+J\le D+2T+G=O(N^2s).
\]
There is therefore no indefinite sequence of interceptions.
As in the original punishment argument, we omit neutralized lateral
royal-lid move--reply pairs: each consumes
one turn per player and changes neither the upper route nor the
remaining royal continuation. The other royal deviations retain
their previously established local responses.

\begin{theorem}
\label{thm:mosquito-free-completeness}
The single-Queen $n$-Hive decision problem and both multi-Queen
variants are EXPTIME-complete without Mosquitoes, even when all pieces are already placed. This holds for every
inventory variant with enough Grasshoppers and Beetles to realize
the construction.
\end{theorem}

\begin{proof}
The replacement preserves faithful simulation and punishes deviations.
Its constant stack heights preserve polynomial size, giving
single-Queen hardness under the stated restrictions.
The same instances give hardness for both multi-Queen variants:
with one Queen per player, their loss conditions coincide.
The previously established EXPTIME membership argument applies
unchanged when Mosquitoes are excluded.
\end{proof}

%% file: Sec/End.tex
\section{Conclusions}
\label{sec:conclusions}

We have proved that deciding whether a player can force a win in
\singleQueenHive{} is EXPTIME-complete, even when all pieces are already
placed. The reduction uses the geometry
and the authentic movement rules of \Hive{} directly: Grasshoppers encode reversible logical
states, the One Hive rule enforces the local constraints, and Mosquitoes turn departures from
the simulation into a global punishment. A polynomial separation between the
prescribed royal traversal and the punishment route completes the construction.
The Beetle-stack replacement of Section~\ref{sec:beetle-stacks} yields the
same complexity result using only Queen Bees, Beetles, and Grasshoppers.

Both multi-Queen variants defined here are also EXPTIME-complete:
the single-Queen reduction applies to either loss condition, and the
same upper bound covers all three models. 

A natural question is whether the reduction can be made more compact.
The current construction dilates the embedding to accommodate royal
regions whose completion time exceeds the global punishment time.
One possible approach would be to use unbounded two-player Constraint
Logic under normal play, with no passes or target edges and with the
first player unable to move losing. Is this variant EXPTIME-hard on
the same planar AND/OR graph family? Such a result would remove the
need for pass gadgets and royal countdowns: a player with no legal
simulated move would be forced to deviate, triggering the opponent's
punishment. This could yield a simpler reduction and substantially
smaller \Hive{} positions.

\section*{Acknowledgements}

We thank the Elicsir Foundation for introducing us to Hive
and giving us the opportunity to witness the first victory
by a bot over the European Hive champion, further fuelling
our shared enthusiasm for the game.

This work also owes something to FUN with Algorithms 2026,
where Andel and Rin's paper was presented. One of us brought a copy of Hive, which was played most evenings of the conference.
The enthusiasm for the game at the conference motivated us to resolve
the question left open in their paper.

We are grateful to the Hive community for the valuable feedback
we received at the 2026 European Hive Convention and Open.
We also thank John Yianni and the Gen42 Games team for creating
such an entertaining game and allowing us to use the official
piece artwork.

Finally, we are especially grateful to Francesco Salerno
for his invitation, warm welcome, and the many long discussions
about our work, Hive bots and playing strategies.